\documentclass[runningheads]{llncs}

\usepackage{graphicx}
\usepackage{url}
\usepackage{amsfonts}
\usepackage{amsmath}
\usepackage{wrapfig}
\usepackage{hyperref}
\usepackage[linesnumbered,ruled]{algorithm2e}
\usepackage[title]{appendix}
\usepackage{makecell}
\usepackage{booktabs, multirow}  
\usepackage{soul}                        
\usepackage[table]{xcolor}                
\usepackage{changepage,threeparttable}    
\usepackage{xcolor}
\usepackage[subtle]{savetrees} % compress space
\usepackage{subcaption}

\renewcommand{\min}{\mathtt{min}}
\renewcommand{\max}{\mathtt{max}}
\newcommand{\argmax}{\mathtt{argmax}}

\newcommand{\phiopt}{\phi_{\mathit{opt}}}
\newcommand{\real}{\mathbb{R}}
\newcommand{\lcf}{\Pi}
\newcommand{\selector}{\phi_k}
\newcommand{\nnctrl}{\pi_{nn}}
\newcommand{\safea}[1]{\bar{a}_{#1}}
\newcommand{\nna}[1]{a^{nn}_{#1}}
\newcommand{\transmat}[1]{\mathcal{T}_{#1}}
\newcommand{\nftransmat}[1]{\mathcal{T}_{\hat{s}_{#1}}}
\newcommand{\nmat}[1]{\mathcal{T}_{w_{#1}}}
\newcommand{\lb}[1]{L({#1})}
\newcommand{\ub}[1]{U({#1})}
\newcommand{\srs}[1]{\mathcal{R}^s_{#1}} % stochastic reachable set 
\newcommand{\rs}[1]{\mathcal{R}_{#1}} % reachable set 
\newcommand{\unsafereg}{\mathcal{S}_u}
\newcommand{\safereg}{\bar{\mathcal{S}}_u}
\newcommand{\rsafe}[1]{r_{\mathit{safe}}(#1)}
\newcommand{\rlive}[1]{r_{\mathit{live}}(#1)}
\newcommand{\dist}[1]{\mathcal{D}_{#1}}
\newcommand{\rpub}{U_f} % reachable probability upper bound
\newcommand{\cumulb}{\mathcal{L}}
\newcommand{\cumulbopt}{\cumulb_{opt}}
\newcommand{\safelb}{\hat{p}_t}
\newcommand{\sss}{\Omega} % selector search space
\newcommand{\ps}{\mathtt{PrefixSame}} 
\newcommand{\initss}{\mathcal{R}_0}
\newcommand{\vnoise}{w'}
\newcommand{\sysnoise}{w}
\newcommand{\region}[2]{\mathtt{Region}(#1, #2)}

\newtheorem{assumption}{Assumption}

\begin{document}
\title{Scalable Synthesis of Verified Controllers\\ in Deep Reinforcement Learning}
\titlerunning{Scalable Synthesis of Verified Controllers in Deep Reinforcement Learning}
 \author{Zikang Xiong \and
   Suresh Jagannathan}
% %
 \authorrunning{Z. Xiong, S. Jagannathan}
 \institute{Purdue University, West Lafayette, IN, USA, 47906 }
\maketitle              % typeset the header of the contribution

\begin{abstract}
  There has been significant recent interest in devising verification techniques for learning-enabled controllers (LECs) that manage safety-critical systems. Given the opacity and lack of interpretability of the neural policies that govern the behavior of such controllers, many existing approaches enforce safety properties through shield, a dynamic monitoring-and-repairing mechanism that ensures a LEC does not emit actions that would violate desired safety conditions. These methods, however, have been shown to have significant scalability limitations because verification costs grow as problem dimensionality and objective complexity increase. In this paper, we propose a new automated verification pipeline capable of synthesizing high-quality safe controllers even when the problem domain involves hundreds of dimensions, or when the desired objective involves stochastic perturbations, liveness considerations, and other complex non-functional properties. Our key insight involves separating safety verification from neural controller training, and using pre-computed verified safety shields to constrain the training process. Experimental results over a range of high-dimensional benchmarks demonstrate the effectiveness of our approach in a range of stochastic linear time-invariant and time-variant systems.

\end{abstract}

\section{Introduction}

Deep Reinforcement Learning (DRL) has proven to be a powerful tool for implementing autonomous controllers for various kinds of cyber-physical systems (CPS). Since these learning-enabled controllers are intended to operate in safety-critical environments, there has been significant recent interest in developing verification methods that ensure their behavior conforms to desired safety properties \cite{alshiekh_safe_2017,anderson2020neurosymbolic,huang2019reachnn,li2020robust,tran2020nnv,zhu2019inductive}. While these different approaches all provide strong guarantees on controller safety, scaling their techniques, both with respect to problem dimensionality as well as objective complexity, has proven to be challenging. Approaches that attempt to verify that a neural controller always preserves desired safety guarantees \cite{huang2019reachnn,liu2021algorithms,tran2020nnv} face challenges to scaling to high dimensions due to the structural complexity of the neural network and increasing over-approximation error as a function of dimensionality. More importantly, the neural controller is typically trained with various performance objectives, in addition to safety~\cite{anderson2020neurosymbolic}. Balancing these competing goals of ensuring safety on the one hand and maximizing objective reward on the other poses its own set of challenges that can compromise verifiability, performance, and safety. Alternatively, a shield framework monitors controller actions and triggers a safety shield when these actions may lead to an unsafe state. By applying safety verification to a simpler linear controller that governs the behavior of this shield, we decouple safety verification from the complexity of the underlying neural network and its objectives, and can thus realize better scalability characteristics. However, a simple linear controller cannot guarantee safety in all scenarios. Thus, a composition on linear controllers is typically required. Compared with previous work~\cite{anderson2020neurosymbolic,yang2021iterative,zhu2019inductive}, which composes linear controllers in a state space, we consider composition in time as shown in Fig.~\ref{fig:overview-ab}(b). Considering composition over time is natural for time-variant systems and our experimental results show that this approach also works for benchmarks in which spatial composition is important~\cite{zhu2019inductive}. 

\begin{figure}[h]
    \centering
    \includegraphics[width=\linewidth]{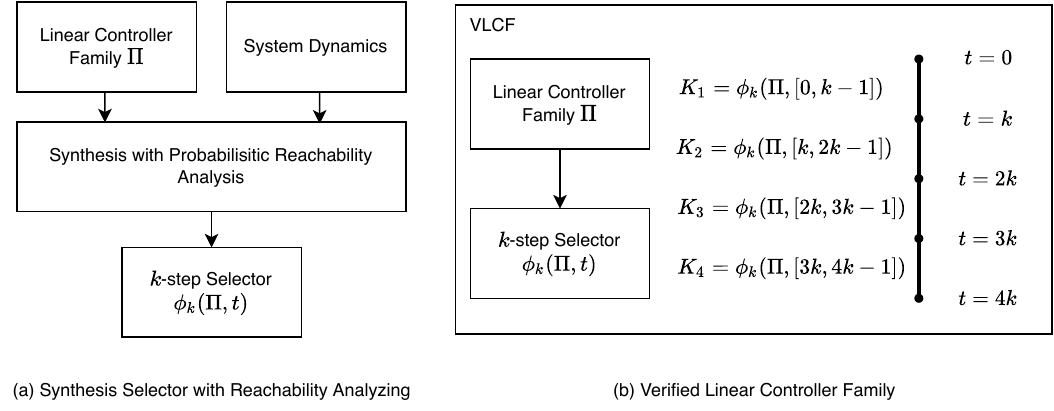}
    \caption{(a) Given a linear controller family $\lcf$ and system dynamics, we synthesize a selector $\selector(\lcf, t)$ to choose the linear controller at time $t$. The chosen linear controller acts as shield at time $t$. (b) A verified controller family contains both $\lcf$ and $\selector(\lcf, t)$; the selector chooses a linear controller $K_i$ every $k$ steps. }
    \label{fig:overview-ab}
\end{figure}

\begin{figure}[t]
    \centering
    \includegraphics[width=\linewidth]{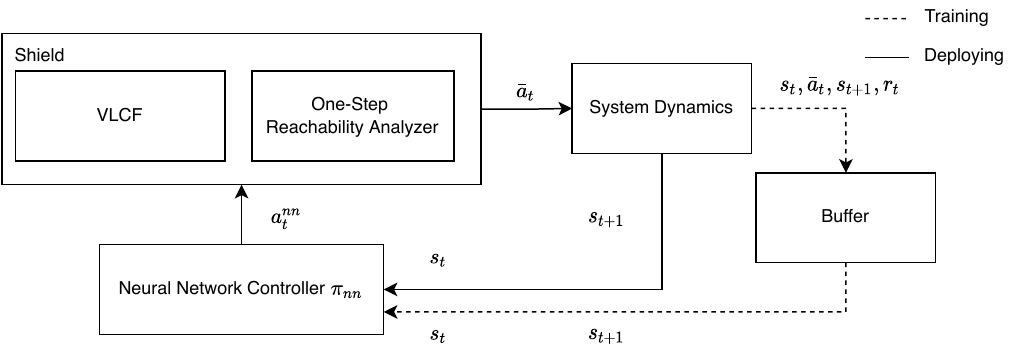}
    \caption{Shielding both training and deploying phases. The shield includes a VLCF and a one-step reachability analyzer. The neural network controller $\nnctrl$ generates an action $\nna{t}$ based on the current state $s_t$, and the one-step reachability analyzer decides whether to intervene and output a safe action $\safea{t}$, where $\safea{t}$ can either be the $\nna{t}$ or a safe action generated by the VLCF. Executing the action $\safea{t}$ generates a new system state $s_{t+1}$ and an immediate reward $r_t$. During training time, the transition $(s_t, \safea, s_{t+1}, r_t)$ is stored in a buffer for training $\nnctrl$ and the new state $s_{t+1}$ is fed to  $\nnctrl$ to predict a new action.  At deployment time, $s_{t+1}$ is fed to $\nnctrl$ directly. The shield monitors safety during both training and deploying phases.}
    \label{fig:shielding_pipeline}
\end{figure}

This paper presents a new learning and verification pipeline that addresses these challenges in stochastic linear time-invariant and time-variant systems. Similar to other shielding-based approaches \cite{alshiekh_safe_2017,anderson2020neurosymbolic,yang2021iterative,zhu2019inductive}, our work does not verify a neural controller directly. Instead, it verifies and composes a family of linear controllers, driven by a novel probabilistic reachability analysis; this verified linear controller family (VLCF) collectively serves as a shield that can dynamically enforce the safe operation of the system. Notably, our technique considers safety verification independently of complex objectives and neural network internal structure, and thus enables scalability with respect to objective complexity and problem dimensionality. In our experiments, for example, we demonstrate successful verification of CPS benchmarks with over 800 dimensions, a scale that is significantly higher than reported with existing approaches. Because the VLCF is generated based only on safety considerations, they are not intended to serve as the primary mechanism for governing the actual operation of the system, which must also take into account other performance-related objectives. We train a neural network controller to achieve these performance-related objectives, while guaranteeing its safety with VLCF. Compared with previous work~\cite{zhu2019inductive} that only considers safety post-training, we also integrate the VLCF into the neural controller training process, thus providing a complete safety-aware pipeline from training to deployment. Our work thus demonstrates a fully automated safety verification pipeline for learning-enabled controllers trained with objectives beyond safety. Our main contributions are as follows:

\begin{itemize}
    \item We propose a verification pipeline for synthesizing and deploying applications involving stochastic linear-time properties, implemented as deep neural network controllers, with both safety and performance objectives.
    \item We present a new verification approach that composes a linear controller family as a safety shield, driven by a novel probabilistic reachability analysis technique. 
    \item We evaluate our methodology on a range of high dimensional stochastic linear time-invariant and time-variant applications, and demonstrate the scalability and effectiveness of our approach.  
\end{itemize}

\section{Preliminaries}
\label{sec:prelim}

\subsection{Controller Types}

We consider two kinds of controllers. We expect the performance controller $\nnctrl$ to be a neural network controller trained by reinforcement learning algorithms. Depending on the algorithm used, $\nnctrl$ can either be deterministic or stochastic. Given a state $s_t \in \real^n$ of the system, $\pi_{nn}(s_t)$ outputs an action $a_t \in \real^m$. A deterministic linear controller family $\lcf$ is a set of linear policies $K_i \in \real^{m \times n}$, where $m$ is the action. In every $k$-unit-time interval $[\lfloor \frac{t}{k} \rfloor k, \lfloor \frac{t}{k} \rfloor k + k)$, a selector $\selector(\lcf, t)$ chooses a linear controller in $\lcf$ for predicting actions in this time interval.

\subsection{Stochastic Linear Transition System}
\label{sec:stochastic_transition_system}

A linear transition system is modeled as
\begin{equation}
  \begin{aligned}
    \left\{\begin{array}{l}\dot{s}_t=A_t \cdot s_t + B_t \cdot a_t \\ s_{t+1} = s_t + \dot{s}_t \Delta t + w \end{array}\right.
  \end{aligned}
  \label{sys}
\end{equation}
where $s_t$ is the state vector and $a_t$ is the action.  Matrices $A_t$ and $B_t$ are two matrices capturing linearized dynamics, and they are used to compute $\dot{s}_t$. $A_t$ and $B_t$ change over time in a time-variant system, and are fixed in time invariant system. Stochasticity is introduced by adding a noise (or error) term $w$ to each transition. In this paper, we consider bounded noise, that is, $w \in [\lb{w}, \ub{w}]$. Although our verification algorithm does not directly support non-linear dynamics, we note that there exists a line of work \cite{tomas2010linear} that shows how to approximate such dynamics using (time-variant) linear systems. Applying these methods to our setting enables generalization of our technique in practice. We also note that approximation error can be considered as part of stochastic noise, thus ensuring soundness even after approximation.

\subsubsection{Reachable Set}

The initial state distribution of a system is bounded $s_0 \in [\lb{s_0}, \ub{s_0}]$.  If we now have a linear controller chosen by $\selector\left(\lcf, 0\right)$, the noise-free state $\hat{s}_1$ at step 1 is
\begin{align*}
  \hat{s}_1 & = s_0 + \Delta t \cdot \left(A_0 s_{0}+B_0 \cdot \selector\left(\lcf, 0\right) s_{0}\right)      \\
            & = \left(I + \Delta t \cdot \left(A_0 + B_0 \cdot \selector\left(\lcf, 0\right)\right)\right) s_0 \\
            & = \transmat{0} s_0
\end{align*}

The state at step 1 $s_1$ after adding the noise term is $\hat{s}_1 + w$ where  $\hat{s}_1$ is the result after a linear transformation on state $s_0$. Let $\transmat{i} = I+\Delta t \cdot\left(A_t + B_t \cdot \selector\left(\lcf, i\right)\right)$. For step $t$, we have $s_{t+1} = \transmat{t} s_t + w$ and thus,
\begin{equation}
  \begin{aligned}
    \label{eq:stochastic_linear_transition}
    s_{t} & = \transmat{t-1}\left(\cdots \left(\transmat{1} \left(\transmat{0} s_{0} + w\right) + w\right) \cdots + w\right) + w \\
          & = \prod_{i=0}^{t-1} \transmat{i} s_{0} + (I + \sum_{i=1}^{t-1}  \prod_{j=i}^{t-1} \transmat{j}) w                    \\
          & = \nftransmat{t} s_{0} + \nmat{t} w
  \end{aligned}
\end{equation}

We denote the stochastic reachable set of a stochastic system at step $t$ as $\srs{t}$, $\nftransmat{t} = \prod_{i=0}^{t-1} \transmat{i}$, and $\nmat{t} = I + \sum_{i=1}^{t-1} \prod_{j=i}^{t-1} \transmat{j}$. $\srs{t}$ can be characterized by a tuple $(\nftransmat{t}, \nmat{t}, \lb{s_0}, \ub{s_0}, \lb{w}, \ub{w})$.

\subsection{Properties and Rewards}
\label{sec: properties and reward}

\subsubsection{Safety}
We require that all the possible reachable sets over time have no intersection with an unsafe region. We define the safe region as $\safereg = \{s | s \in [\lb{\safereg}, \ub{\safereg}] \}$. An unsafe region $\unsafereg$ is the complement of $\safereg$. Therefore, our desired safety property requires that $\forall t, \srs{t} \cap \unsafereg = \emptyset$. For any state $s \in \mathbb{R}^n$, we encode the safety property as a reward $\rsafe{s}$,
\begin{align}
  \rsafe{s} = \sum_{i=0}^{n-1} \left(\min\left( \max (s - \lb{\safereg}), 0\right) + \min \left(\max (\ub{\safereg} - s), 0\right)\right).
\end{align}
Since $s$, $\lb{\safereg}$, $\ub{\safereg}$ are all vectors, the function $\max$ computes the max value on every element of a vector. Any time one dimension of the state does not stay in the safe region $\safereg$, $\rsafe{s} < 0$, a penalty is ascribed that decreases the likelihood of visiting this state when training a neural network controller. 

\subsubsection{Liveness}
Our desired liveness property requires system states to keep changing. For example, we expect a robot to keep moving towards its goal and to remain above a certain speed. In this case, we can specify that the robot's speed should always be greater than a value. If the robot's speed drops below the specified value, its behavior should be penalized.  Other examples of such liveness properties may require the rate at which a tank is filled to exceed the rate at which it is emptied; a LIDAR sensor should always rotate, etc. To characterize such patterns, we define the following reward:
\begin{align*}
  \rlive{s} = \sum_{i \in \mathtt{Dim}_{\mathit{live}}} f_>(|s^{(i)}|, T^{(i)}),
\end{align*}
where
$$
  f_>(a, b) = \left\{\begin{array}{cl}
    1, & a > b \\
    0, & else
  \end{array}\right. .
$$
$T \in \mathbb{R}^n$ is a vector of thresholds. $s^{(i)}$ and $T^{(i)}$ are $i$-th element of $s$ and $T$, respectively. $\mathtt{Dim}_{\mathit{live}}$ contains all the state dimensions we want to check (e.g., the speed of a moving robot). If the absolute value of one state dimension is greater than this value, we give a positive reward. Maximizing this reward means we want as many dimensions as possible to hit the threshold.

\section{Approach}
\label{sec:approach}

\subsection{Generate Linear Controller Family}
\label{sec:generate_lcf}

  We generate a linear controller family $\lcf$ using a \emph{Linear–Quadratic Regulator}(LQR) w.r.t different cost functions and perturbed dynamics. This is different from previous works~\cite{anderson2020neurosymbolic,yang2021iterative,zhu2019inductive} that distill the linear controller family from a neural network policy. Distilling only works well when the neural network controller's objective is aligned with the safety objective. However, when the objectives become complex, they often fail to generate a verifiable linear controller, as we show in \ref{sec: pendulum_distill_challenging}. An LQR controller, on the other hand, is designed to stabilize a system, and thus generates good controllers stabilized around safe states.

  Diversity in $\lcf$ is desired because we hope to compose different linear controllers for different scenarios. For example, a drone operating in windy conditions can experience wind coming from different directions requiring different controllers to be involved to stabilize its actions according to the wind direction at a given state. We generate these controllers by perturbing the LQR cost function and the dynamics matrices $A_t$ and $B_t$.

  LQR computes an optimal linear controller by minimizing cost $J(\tau)$ over trajectory $\tau$, which is generated from a linear transition system. The cost function is
  \begin{align}
    \label{eq:lqr_cost}
    J(\tau) = \sum_{s, a \sim \tau} s^\top Q s + a^\top R a.
  \end{align}
  We randomly perturb $Q$ and $R$ to generate different linear controllers. Similarly, we can also perturb $A_t$ and $B_t$ to reflect the change of dynamics and thus generate different controllers in different scenarios. When generating linear controller family $\lcf$, we simultaneously perturb $Q, R, A_t, B_t$, which diversifies controllers in $\lcf$.

\subsection{Safety Probability of Stochastic Reachable Sets}
\label{sec:safety_probability_of_reachable_sets}
We define the safety probability of a stochastic reachable set $\srs{t}$ in this section and introduce some theoretical results that underpin our reachability analysis. Proofs for these theoretic results are detailed in \ref{sec:prob_reach_analisis}.

Suppose $f_t(s)$ is the PDF of state $s_t$. The initial state is $s_0 \sim \dist{s_0}$, where $\dist{s_0}$ is a uniform distribution. The noise is sampled from a distribution, $w \sim \dist{w}$. We assume noise is bounded (i.e., it has fixed upper and lower bounds). According to Eq.~\ref{eq:stochastic_linear_transition}, the distribution of reachable states at any given step is the linear combination of $\dist{s_0}$ and $\dist{w}$. Given a reachable set $\srs{t}$, we wish to characterize a
distribution of its safety probability $p_t$, a measure that indicates
the portion of $\srs{t}$'s ``surface'' that is safe:

\begin{align*}
  p_t(s) = \int_{s \in {\safereg \cap \srs{t}}} f_{t}(s)\ ds
\end{align*}

\begin{theorem}
  \label{trm:upper_bound}
  Suppose that $s_0$ is subject to a uniform distribution on $\left[\lb{s_0}, \ub{s_0}\right]$. Let $\delta=\ub{s_0}-\lb{s_0}$,
  \begin{align*}
    f_{t}(s) \leq \frac{1}{\left|\det\left(\nftransmat{t} \right)\right| \prod_{i=0}^{n-1} \delta_{i}}
  \end{align*}
\end{theorem}
where $\text{det}$ is the determinant of a matrix. Since the initial state $s_0$ is subject to a uniform distribution, the area of the initial reachable set is $\prod_{i=0}^{n-1} \delta_{i}$. The determinant of a linear transformation represents how much of the initial area scales after the linear transformation. Thus, $\left|\det\left(\nftransmat{t}\right)\right| \prod_{i=0}^{n-1} \delta_{i}$ computes the reachable set's area at step $t$. This theorem asserts that the probability density function $f_t(s)$ is bounded above by the reciprocal value of the area of a reachable set. Intuitively, if the area of a reachable set is large, the probability density is stretched to be small.

\begin{corollary}
  \label{trm:uniform_noise_upper_bound}
  If noise is subject to a uniform distribution, and its distribution is on $\left[\lb{w}, \ub{w}\right]$, $\delta^{\prime}=\ub{w} - \lb{w}$, then
  \begin{align}
    f_{t}(s) \leq \min \left(\frac{1}{\left|\det\left(\nftransmat{t}\right)\right| \prod_{i=0}^{n-1} \delta_{i}}, \frac{1}{\left|\det\left(\mathcal{T}_{w_t}\right)\right| \prod_{i=0}^{n-1} \delta_{i}^{\prime}}\right)
  \end{align}
\end{corollary}
Corollary~\ref{trm:uniform_noise_upper_bound} says when the noise is uniform, the probability density can be further bounded. Intuitively, large uniform noise makes the reachable set cover a larger area. Hence, the PDF is stretched to be smaller. When the noise is significant, $\prod_{i=0}^{n-1} \delta_{i}^{\prime}$ becomes large. As a result, as $\frac{1}{\left|\det\left(\mathcal{T}_{w_t}\right)\right| \prod_{i=0}^{n-1} \delta_{i}^{\prime}}$ becomes tighter, the upper bound of $f_t(s)$ becomes smaller.  Proofs of Theorem~\ref{trm:upper_bound} and Corollary~\ref{trm:uniform_noise_upper_bound} are given in \ref{apdx:upper_bound_of_pdf}.

Suppose the upper bound of $f_t(s)$ is $\rpub$. In this case, the intersection between the reachable set $\srs{t}$ and the \emph{unsafe state set} $\unsafereg$ is ${\unsafereg \cap \srs{t}}$.  The safety probability $p_t$ satisfies
\begin{align}
  p_t \geq 1 - \rpub \int_{S_u \cap \srs{t}} ds
\end{align}
We denote the safe lower bound $\hat{p}_t = 1 - \rpub \int_{S_u \cap \srs{t}} ds$.\footnote{In \ref{apdx:reachable_approximation}, we provide a computationally efficient method to compute the upper bound of $\int_{S_u \cap \srs{t}} ds$. \ref{apdx:prob_lb} provides an efficient approach to compute $\rpub$.} Let the cumulative safety lower bound be $\cumulb = \sum_{i=1}^M \safelb$. When $\cumulb = M$, the system is verified to be safe in $M$ steps. $\safelb$ is parameterized as $\safelb\left(\lcf, \selector\right)$, because generating $\srs{t}$ depends on $\lcf$ and $ \selector$.

\subsection{Synthesize Selector}
\label{sec:synthesis_selector}

In this section, we introduce a synthesis process for the selector $\selector$ over a linear controller family $\lcf$. The synthesis algorithm checks the reachable sets of every single step w.r.t different selectors (i.e., different linear controller compositions over time), computes their safety probability lower bound $\safelb$, and returns the selector with the highest cumulative safety probability lower bound $\sum_{t=1}^M \safelb$. The sketch of the synthesis process is shown in Algorithm~\ref{algo:verify_linear_controller_family_sketch}, and the full algorithm is provided in the appendix (Algorithm~\ref{algo:synthesis_phi_opt}). The input of our algorithm includes the maximum number of execution steps $M$, linear controller family $\lcf$, and the search space of the selectors $\sss$. The main body of the algorithm iterates the selector in the search space $\sss$ while cutting the search space at runtime. In lines 2-3, we compute the stochastic reachable set $\srs{t}$ described in Sec.~\ref{sec:stochastic_transition_system} and the safety probability lower bound in Sec.~\ref{sec:safety_probability_of_reachable_sets}. Line 4-6 are three strategies to cut the search space $\sss$. We provide intuitive demonstrations of the key components and the three cutting strategies in Fig~\ref{fig:verification_demonstration}.

\begin{algorithm}[!ht]
  \caption{\small Sketch of Algorithm~\ref{algo:synthesis_phi_opt}}
  \label{algo:verify_linear_controller_family_sketch}
  \SetAlgoLined
  \KwIn{$M, \lcf, \sss$}
  \KwOut{$\phiopt, \cumulbopt$}

  \For{$\selector \in \sss$}{
    Compute the stochastic reachable set $\srs{t}$; \\
    Compute the cumulative safety lower bound $\mathcal{L}$; \\
    Cut $\sss$ with $\cumulb$ and update $\cumulbopt$. \\
    Cut $\sss$ with relationship between selectors \\
    Cut $\sss$ with invariant of reachable set \\
  }
\end{algorithm}

The search space $\sss$ is a tree because we select different controllers in $\lcf$ every $k$ steps. Each node in the tree represents one controller selected from step $(d - 1) \times k$ to step $d \times k - 1$, where $d$ is the depth of the corresponding tree node. We define a function $\ps(\selector, m)$ that represents a set of selectors that have the same ancestor with $\selector$, until depth (step) $m$. This function is beneficial because we can often cut a set of selectors with the same ancestor from selector search space $\sss$.

\begin{figure}[ht]
  \centering
  \includegraphics[width=\linewidth]{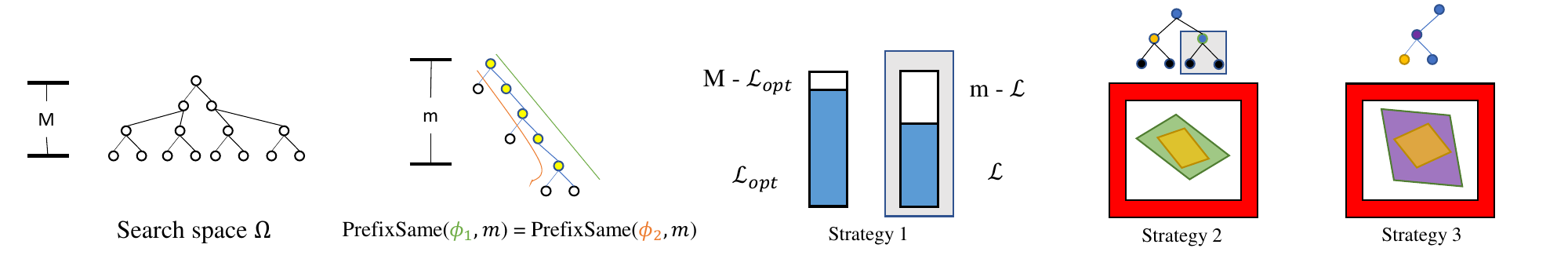}
  \caption{\small Demonstration of the key components for selector synthesis algorithm}
  \label{fig:verification_demonstration}
\end{figure}

Three strategies reduce the search space during controller selection. Strategy 1 is based on the best cumulative safety probability $\cumulb_{opt}$ recorded. We use $m - \mathcal{L}$ to represent the cumulative \emph{unsafe} probability upper bound, where $m$ is the time step of the current run, and $\mathcal{L}$ is the cumulative safety probability lower bound until step $m$. When we find that $m - \mathcal{L}$ is greater than the best cumulative unsafe probability upper bound $M - \cumulbopt$, we will not find a better cumulative safety lower bound as steps increase. Suppose that $\selector$ results in $\mathcal{L}$; we can now cut selectors $\mathtt{PrefixSame}(\selector, m)$. Strategy 2 shrinks the search space by comparing the reachable set between two selectors. For example, in Strategy 2 of Fig~\ref{fig:verification_demonstration}, the reachable set of the yellow and green nodes are colored as yellow and green, respectively. The yellow reachable set is the subset of the green one. Thus, selecting the yellow node is strictly safer than selecting the green. In this case, all the selectors that share the green node as an ancestor can be removed. Strategy 3 computes the invariant for a single selector. For example, in Strategy 3 of Fig~\ref{fig:verification_demonstration}, the yellow node and the purple node belong to the same selector. The yellow node's reachable set is a subset of the purple node's reachable set. Thus, the reachable set of this selector will shrink over time. As a result, all the reachable sets will always be the subset of the largest reachable set until the step of purple nodes. The subset relationship hence serves as an invariant. If this largest reachable set has no intersection with unsafe region $\unsafereg$, the reachable set will never intersect with the unsafe region, and thus we can directly return a verified selector.

\subsection{Shield}
\label{sec: safe_training_with_shield}

\begin{figure}
  \includegraphics[width=\linewidth]{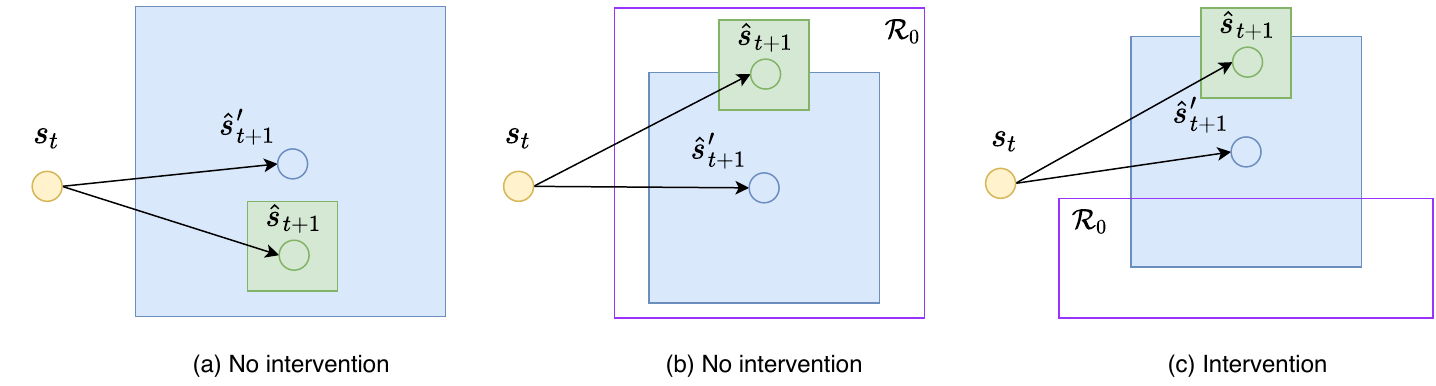}
  \caption{One-step reachability analysis of the shield. (a) The possible reachable set is inside the verified reachability set. (b) The possible reachable set slides out the verified reachability set but stays in the initial state space. (c) The possible reachable set slides out of both the verified reachability set and the initial state space. Cases (a) and (b) do not require intervention, but (c) does.}
  \label{fig:one-step-analysis}
  \vspace{-20px}
\end{figure}

  The shield has two components - the VLCF and a one-step reachability analyzer as shown in Fig.~\ref{fig:shielding_pipeline}. Integrating the linear controller family (Sec.~\ref{sec:generate_lcf}) with the synthesized selector (Sec.~\ref{sec:synthesis_selector}), gives us the VLCF. We demonstrate the one-step reachability analysis process in Fig.~\ref{fig:one-step-analysis}. The analyzer looks one step ahead of the current state. The blue box is the verified reachability set generated by the VLCF. If we take over control using the VLCF from any state of this reachable set, \emph{all states} are guaranteed to be safe over \emph{all the future steps} until max simulation length $M$, with high probability. The green box is the stochastic reachability set after taking the action generated by the neural network controller. The purple box $\initss$ is the initial state space. All the states in this set are guaranteed to be safe in all future steps under the control of VLCF. Fig.~\ref{fig:one-step-analysis} depicts 3 scenarios that a one-step analyzer can encounter. The shield only needs to intervene in the case shown in Fig.~\ref{fig:one-step-analysis}(c).

\begin{algorithm}[ht]
  \caption{\small Shield algorithm for $\nnctrl$}
  \label{algo:shield}
  \KwIn{$\nnctrl, \lcf, \phiopt, \vnoise, \sysnoise, s_t$}
  \KwOut{safe action $\safea{t}$}
  $a_{nn} \leftarrow \nnctrl(s_t)$\;
  $K \leftarrow \phiopt(\lcf)$\;
  $a_{\mathit{safe}} \leftarrow K \cdot s_t$\;
  $\hat{s}_{t+1} = dynamics(s_t, a_{nn})$\;
  $\hat{s}'_{t+1} = dynamics(s_t, a_{\mathit{safe}})$\;
  \uIf{$\mathtt{Region}(\hat{s}_{t+1}, w) \subseteq \left(\mathtt{Region}(\hat{s}'_{t+1}, w') \cup \mathcal{R}_0\right)$}{
    \KwRet $a_{nn}$\;
  } \Else{
    \KwRet $a_{\mathit{safe}}$\;
  }
\end{algorithm}

Algorithm \ref{algo:shield} summarizes the operation of a shield. Given a neural network controller $\nnctrl$, a linear controller family $\lcf$, synthesized selector $\phiopt$, the verified noise $\vnoise$, the real noise of system $\sysnoise$, and a state $s_t$, Algorithm~\ref{algo:shield} returns a shielded action $\safea{t}$. $\region{s}{w}$ in line 6 is a function computing all possible states allowed by given state $s$ and noise $w \in [\lb{w}, \ub{w}]$. $\region{s}{w}= \{s' | (s' - s) \in [\lb{w}, \ub{w}]\}$. We apply Algorithm~\ref{algo:shield} to enforce the chosen action is safe.  A particularly important instance of this approach is when $\sum_{i=1}^M \hat{p}_t = M$ as formalized by the Theorem~\ref{trm:soundness_of_shield}.

\begin{theorem}
  \label{trm:soundness_of_shield}
  (Soundness of Shield) If $\sum_{i=1}^M \hat{p}_t = M$,  $s_0 \in \mathcal{R}_0$, and $\forall t < M$, $\safea{t}$ is generated by Algorithm~\ref{algo:shield}, then
  $\forall t \leq M, s_t \notin S_u$.
\end{theorem}

\begin{proof}
  See \ref{sec: soundness_of_shield}
\end{proof}

\section{Experimental Results}
\label{sec:expt}
We have applied our verification strategy to various stochastic transition systems, whose dimensions range from 2 to 896. We associate safety constraints with each benchmark, synthesizing a verified linear controller family that seeks to guarantee these properties hold, and use that family to train a neural network with additional performance (a.k.a. liveness) objectives. The resulting system consists of a performance-sensitive neural network trained with the awareness of safety constraints, coupled with a safety shield represented by the linear controller family. We train the neural network controller using Proximal Policy Optimization (PPO)\cite{schulman1707proximal}, a widely-used training technique used in reinforcement learning.

\subsubsection{Benchmarks}
We evaluate our algorithm on 24 benchmarks. There are 6 base benchmarks - {\sf Pendulum} {\sf Cartpole}, {\sf Cartpole}, {\sf Carplatoon}, and {\sf Helicopter}, {\sf DroneInWind}. The {\sf DroneInWind} environment is time-variant because we allow the angle and the strength of the wind in the environment to change over time; the other benchmarks are time-invariant. We also consider stacking environment variants of these benchmarks named $n$-{\sf B} for assessing the effectiveness of our approach as dimensionality increases; here, $n$ is the stacking depth, and {\sf B} is one of the six base benchmarks. To make the experiments not simply exploit the safety characteristics discovered for the base program, each stacked layer is defined with a randomly injected offset that makes every stacked element different from every other one. The details of these benchmarks are provided in \ref{sec:benchmark_details}.

\subsubsection{Safe Training and Performance after Deploying}
We train a neural controller using the safety guarantees captured by our verified linear controller family. Table~\ref{tab:safe_training} demonstrates the effectiveness of our approach on overall network performance. Although the LQR controller is verified to be safe, it can perform poorly when performance objectives are taken into account. However, using it as a shield for a performant neural controller can realize both performance and safety benefits. The comparison between the rewards of the different controllers is summarized in Table~\ref{tab:safe_training}. The table presents the performance characteristics of the shielded controller relative to the base PPO algorithm without augmentation of a safety planner and the LQR family that is implemented without performance objectives. Numbers greater than one in the column labeled {\sf Shield/PPO} indicate that the controller trained in conjunction with the safety shield outperformed the PPO-only trained algorithm. A similar interpretation holds for the column labeled {\sf Shield/LQR}. The Vio. in Training column indicates the number of safety violations encountered during training - PPO trained networks exhibited a non-negligible number of safety violations on every benchmark; since our verification algorithm was able to generate a provably safe shield for each benchmark, the safety-augmented controller exhibited no violation in any of the benchmarks. In the Perf. after Deploying column, while it is not surprising that controllers trained with both safety and performance (Shield) would outperform those that are only aware of safety (LQR), it is notable that the shielded controller has a higher performance reward than the PPO-trained controller on 19 of the 24 benchmarks.

\begin{table}[thb]
  \centering
  \scriptsize
  \begin{tabular}{l|cc|cc|ccr}\toprule
    \multirow{2}{*}{Benchmarks} & \multicolumn{2}{c|}{Dimension } & \multicolumn{2}{c|}{Vio. in Training $\downarrow$} & \multicolumn{2}{c}{Perf. after Deploying $\uparrow$}                                                \\\cmidrule{2-7}
                                & State                           & Action                                            & PPO Vio.                                             & Shield Vio. & Shield/PPO    & Shield/LQR     \\\midrule
    {\sf Pendulum}              & 2                               & 1                                                 & 1437                                                 & \textbf{0}  & \textbf{2.65} & \textbf{8.59}  \\
    {\sf Cartpole}              & 4                               & 1                                                 & 959                                                  & \textbf{0}  & \textbf{1.36} & \textbf{3.54}  \\
    {\sf DroneInWind}           & 6                               & 2                                                 & 864467                                               & \textbf{0}  & \textbf{3.35} & \textbf{4.54}  \\
    {\sf Carplatoon}            & 15                              & 8                                                 & 69                                                   & \textbf{0}  & \textbf{1.83} & \textbf{30.58} \\
    {\sf Oscillator}            & 18                              & 2                                                 & 3                                                    & \textbf{0}  & \textbf{1.37} & \textbf{6.79}  \\
    {\sf Helicopter}            & 28                              & 6                                                 & 30                                                   & \textbf{0}  & \textbf{1.04} & \textbf{1.49}  \\ \midrule
    {\sf 2-Pendulum}            & 4                               & 2                                                 & 2375                                                 & \textbf{0}  & \textbf{1.77} & \textbf{2.84}  \\
    {\sf 2-Cartpole}            & 8                               & 2                                                 & 1775                                                 & \textbf{0}  & 0.64          & \textbf{2.59}  \\
    {\sf 2-DroneInWind}         & 12                              & 4                                                 & 863053                                               & \textbf{0}  & \textbf{3.14} & \textbf{4.00}  \\
    {\sf 2-Carplatoon}          & 30                              & 16                                                & 1137                                                 & \textbf{0}  & 0.76          & \textbf{11.84} \\
    {\sf 2-Oscillator}          & 36                              & 4                                                 & 46                                                   & \textbf{0}  & \textbf{1.18} & \textbf{3.33}  \\
    {\sf 2-Helicopter}          & 56                              & 12                                                & 277                                                  & \textbf{0}  & \textbf{1.07} & \textbf{1.48}  \\  \midrule
    {\sf 4-Pendulum}            & 8                               & 4                                                 & 4736                                                 & \textbf{0}  & \textbf{2.68} & \textbf{2.34}  \\
    {\sf 4-Cartpole}            & 16                              & 4                                                 & 3529                                                 & \textbf{0}  & 0.60          & \textbf{2.21}  \\
    {\sf 4-DroneInWind}         & 24                              & 8                                                 & 1748560                                              & \textbf{0}  & \textbf{3.61} & \textbf{4.08}  \\
    {\sf 4-Carplatoon}          & 60                              & 32                                                & 1863                                                 & \textbf{0}  & 0.60          & \textbf{15.94} \\
    {\sf 4-Oscillator}          & 72                              & 8                                                 & 150                                                  & \textbf{0}  & \textbf{2.20} & \textbf{3.77}  \\
    {\sf 4-Helicopter}          & 112                             & 24                                                & 405                                                  & \textbf{0}  & \textbf{1.17} & \textbf{1.33}  \\  \midrule
    {\sf 8-Pendulum}            & 16                              & 8                                                 & 11305                                                & \textbf{0}  & \textbf{1.28} & \textbf{1.93}  \\
    {\sf 8-Cartpole}            & 32                              & 8                                                 & 12680                                                & \textbf{0}  & \textbf{1.17} & \textbf{2.55}  \\
    {\sf 8-DroneInWind}         & 48                              & 16                                                & 3551103                                              & \textbf{0}  & \textbf{2.11} & \textbf{2.19}  \\
    {\sf 8-Oscillator}          & 144                             & 16                                                & 579                                                  & \textbf{0}  & \textbf{1.99} & \textbf{3.29}  \\
    {\sf 8-Helicopter}          & 224                             & 48                                                & 1388                                                 & \textbf{0}  & \textbf{1.33} & \textbf{1.24}  \\
    \bottomrule
  \end{tabular}
  \caption{\small Effectiveness of the Shield in Training and Deploying Phases }
  \label{tab:safe_training}
  \vspace{-40px}
\end{table}

\subsubsection{Verification Results}
\label{sec:detailed_verification_results}

\begin{table}[thb]
  \centering
  \setlength{\belowcaptionskip}{-5pt}
  \scriptsize
  \begin{tabular}{@{}l|cccccc|cc@{}}
    \toprule
    {\sf Benchmarks}    & \thead{{\sf state}                                                                             \\ {\sf dim}} & \thead{{\sf action} \\ {\sf dim}} & $M$ & $k$ & $w'$ & {\sf ver. time per} $\lcf$ & {\sf total ver. time} \\\midrule
    {\sf Pendulum}      & 2                  & 1  & 500  & 100 & 1.5e-2 & 0.54s $\pm$ 0.05s     & 0.82s $\pm$ 0.33s      \\
    {\sf Cartpole}      & 4                  & 1  & 500  & 100 & 3e-3   & 0.80s $\pm$ 0.26s     & 1.24s $\pm$ 0.45s      \\
    {\sf DroneInWind}   & 6                  & 2  & 1000 & 100 & 2.5e-3 & 1.39s $\pm$ 0.21s     & 1.49s $\pm$ 0.31s      \\
    {\sf Carplatoon}    & 15                 & 8  & 1000 & 100 & 2e-3   & 1.07s $\pm$ 0.25s     & 3.14s $\pm$ 4.26s      \\
    {\sf Oscillator}    & 18                 & 2  & 1000 & 100 & 4e-3   & 0.69s $\pm$ 0.08s     & 0.72s $\pm$ 0.21s      \\
    {\sf Helicopter}    & 28                 & 6  & 1000 & 100 & 2e-3   & 1.36s $\pm$ 0.42s     & 2.14s $\pm$ 0.86s      \\\midrule
    {\sf 2-Pendulum}    & 4                  & 2  & 500  & 100 & 1.5e-2 & 0.75s $\pm$ 0.20s     & 1.04s $\pm$ 0.24s      \\
    {\sf 2-Cartpole}    & 8                  & 2  & 500  & 100 & 3e-3   & 0.81s $\pm$ 0.19s     & 1.26s $\pm$ 0.69s      \\
    {\sf 2-DroneInWind} & 12                 & 4  & 1000 & 100 & 2.5e-3 & 1.62s $\pm$ 0.29s     & 1.71s $\pm$ 0.34s      \\
    {\sf 2-Carplatoon}  & 30                 & 16 & 1000 & 100 & 2e-3   & 1.32s $\pm$ 0.32s     & 2.44s $\pm$ 3.53s      \\
    {\sf 2-Oscillator}  & 36                 & 4  & 1000 & 100 & 4e-3   & 0.83s $\pm$ 0.14s     & 0.99s $\pm$ 0.45s      \\
    {\sf 2-Helicopter}  & 56                 & 12 & 1000 & 100 & 2e-3   & 4.37s $\pm$ 3.92s     & 5.54s $\pm$ 6.76s      \\\midrule
    {\sf 4-Pendulum}    & 8                  & 4  & 500  & 100 & 1.5e-2 & 0.56s $\pm$ 0.06s     & 0.62s $\pm$ 0.23s      \\
    {\sf 4-Cartpole}    & 16                 & 4  & 500  & 100 & 3e-3   & 0.97s $\pm$ 0.20s     & 1.94s $\pm$ 1.00s      \\
    {\sf 4-DroneInWind} & 24                 & 8  & 1000 & 100 & 2.5e-3 & 1.66s $\pm$ 0.31s     & 1.96s $\pm$ 0.51s      \\
    {\sf 4-Carplatoon}  & 60                 & 32 & 1000 & 100 & 2e-3   & 1.17s $\pm$ 0.20s     & 1.83s $\pm$ 0.61s      \\
    {\sf 4-Oscillator}  & 72                 & 8  & 1000 & 100 & 4e-3   & 1.34s $\pm$ 0.27s     & 1.52s $\pm$ 0.39s      \\
    {\sf 4-Helicopter}  & 112                & 24 & 1000 & 100 & 2e-3   & 8.30s $\pm$ 3.33s     & 8.28s $\pm$ 3.99s      \\\midrule
    {\sf 8-Pendulum}    & 16                 & 8  & 500  & 100 & 1.5e-2 & 0.97s $\pm$ 0.13s     & 0.94s $\pm$ 0.47s      \\
    {\sf 8-Cartpole}    & 32                 & 8  & 500  & 100 & 3e-3   & 1.18s $\pm$ 0.33s     & 1.45s $\pm$ 0.82s      \\
    {\sf 8-DroneInWind} & 48                 & 16 & 1000 & 100 & 2.5e-3 & 2.08s $\pm$ 0.22s     & 2.38s $\pm$ 0.30s      \\
    {\sf 8-Carplatoon}  & 120                & 64 & 1000 & 100 & 2e-3   & 4.03s $\pm$ 0.75s     & 5.06s $\pm$ 1.86s      \\
    {\sf 8-Oscillator}  & 144                & 16 & 1000 & 100 & 5e-3   & 5.78s $\pm$ 1.58s     & 38.60s $\pm$ 66.35s    \\
    {\sf 8-Helicopter}  & 224                & 48 & 1000 & 100 & 2e-3   & 45.03s $\pm$ 31.45s   & 70.11s $\pm$ 92.06s    \\\midrule
    {\sf 16-Helicopter} & 448                & 16 & 1000 & 100 & 2e-3   & 164.05s$\pm$ 36.52s   & 458.17s$\pm$ 343.75s   \\
    {\sf 32-Helicopter} & 896                & 32 & 1000 & 100 & 2e-3   & 2115.20s$\pm$ 1090.29 & 2962.55s$\pm$ 2907.33s \\
    {\sf 64-Helicopter} & 1792               & 64 & 1000 & 100 & 2e-3   & TO                    & TO                     \\
    \bottomrule
  \end{tabular}
  \caption{\small Verification results with timeout set to 1 hour. The {\sf state dim} and {\sf action dim} columns denote the number of dimensions in the state and action space of the benchmark, respectively.  Stacking system benchmarks $n$-* have $n$ times the number of dimensions as their single system counterpart. $M$ is the number of execution time steps considered. The time interval between choosing a new linear controller is $k$. We use $w'$ to represent the verified max noise for each step in this system. The {\sf verification time per} $\lcf$ is the time for verifying a single linear controller family, while {\sf total ver. time} indicates the time to find a controller combination that is guaranteed to be safe (i.e., $\sum_{t=1}^M \hat{p}_t = M$).}
  \label{tab:verification}
\end{table}

The verification results for our benchmarks are shown in Table~\ref{tab:verification}. Although our system supports safety guarantees with probabilistic bounds shown in Table~\ref{tab:verification} were fully verified. In ~\ref{sec:probabilistic_guarrantee_exps}, we present results that verify safety under probabilistic guarantees; these experiments require increasing system stochasticity (term $\omega$ in Equation~\ref{sys}). In experiments of Table~\ref{tab:verification}, the probability safety lower bound was used to prune the search space explored by Algorithm~\ref{algo:synthesis_phi_opt}. The noise term $w'$ is the most extensive noise term we verified.

% To achieve $100\%$ safety, we iteratively run Algorithm~\ref{algo:synthesis_phi_opt} on different linear controller families. Because the linear controller family is generated with randomly sampled $Q$ and $R$ values, the quality of its controllers is hard to guarantee. Thus, if we cannot find a verified selector which maximizes the sum of $p_t$ to $M$ after checking at most 100 possible choices for $\selector$, we regenerate a linear controller family by sampling new $Q$ and $R$ values and re-running Algorithm~\ref{algo:synthesis_phi_opt} on it again.  For each run, if the controller is not fully verified after we check 100 possible choices for $\selector$, we repeatedly generate a new LQR controller family and run  Algorithm~\ref{algo:synthesis_phi_opt} until we get a fully verified linear combination.

We run the verification algorithm 10 times on each benchmark. For each run, if the controller is not fully verified after we check 1000 possible choices for $\selector$, we repeatedly generate a new LQR controller family and run  Algorithm~\ref{algo:synthesis_phi_opt} until we get a fully verified linear combination. The number before $\pm$ in the last 2 columns signifies the mean of our results - we run the verification algorithm 10 times for each benchmark; the number after $\pm$ is the standard deviation. The verification time per $\lcf$ column contains the running time of generating an LQR linear controller family with 10 potential controllers and running Algorithm~\ref{algo:synthesis_phi_opt} once.

Dimensionality is not the only feature that affects verification time. Different safety properties and system dynamics can also play a role here. For example, {\sf 8-Carplatoon} has 120 dimensions but only requires 4.03 seconds to verify on average, while the {\sf 4-Helicopter} benchmark with 112 dimensions requires 8.30 seconds to verify. For {\sf Pendulum} and its stacked systems, verification time per linear controller family is close to total verification time, implying that there was little need to regenerate new controller instantiations. For more complicated benchmarks such as {\sf 8-Oscillator} and {\sf 8-Helicopter}, the linear controller families needed to be regenerated more often, increasing the total verification time. Nonetheless, verification times, even for challenging benchmarks like {\sf 32-Helicopter} with 896 dimensions, required less than 1 hour on average; {\sf 8-Helicopter} with 224 required 70.11 secs to verify on average. The largest benchmark we can verify in the one-hour time limit is the {\sf 32-Helicopter} with 896 state dimensions. However, we did not include {\sf 16-Helicopter} and {\sf 32-Helicopter} in the safe training experiment in Table~\ref{tab:safe_training}. This is because such high dimension models are challenging for a deep reinforcement learning algorithm to find a reasonable controller in one hour, and do not reflect a limitation of our methodology.

\subsubsection{Synthesis Time and Comparison}

We compared our work with other shield-based approaches such as~\cite{zhu2019inductive,yang2021iterative}, where the authors also verified a linear controller family with the barrier-certificate-based approach and a counter-example guided inductive synthesis (CEGIS) loop. The barrier-certificate-based approach is widely used for polynomial dynamics. However, in the stochastic linear dynamic system that we are analyzing, its scalability is limited. We compared our verification algorithm with the tool provided in~\cite{zhu2019inductive} on 20 time-invariant benchmarks. The results are presented in Fig.~\ref{fig:compare_with_barrier_certificate}. Our algorithm is significantly faster than their barrier-certificate-based approach. On 13 of the 20 benchmarks, their tool was unable to find verified controllers within a one-hour time limit. \cite{sreach} supports stochastic and potentially time-variant systems. However, it was only able to verify a single controller from this benchmark set. As \cite{zhu2019inductive,yang2021iterative,anderson2020neurosymbolic} pointed out, in a learning-enabled system, a single verified controller is usually not sufficient to build a shield to guard the safety of the entire state space.

\begin{figure}
  \centering
  \includegraphics[width=0.8\linewidth]{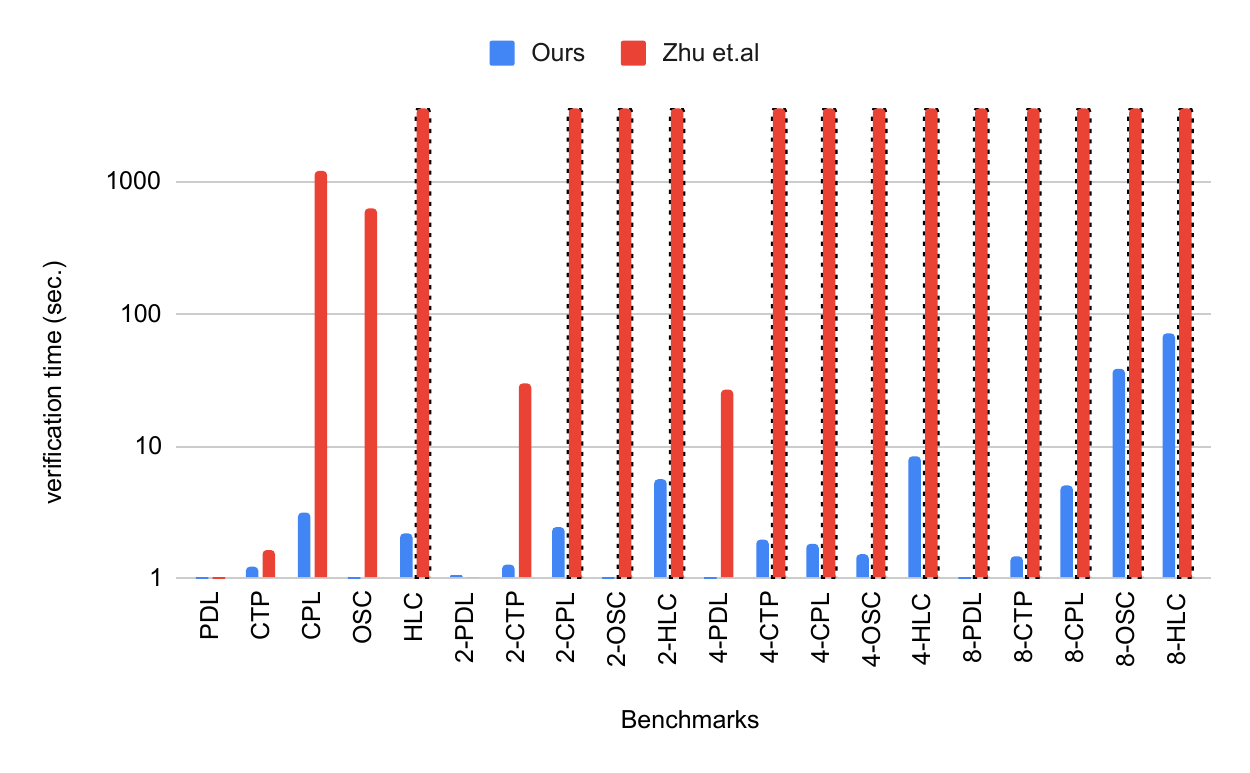}
  \caption{\small Comparison of synthesis time between our work and \cite{zhu2019inductive,yang2021iterative}. The black dash frame means that verifier is unable to return a feasible solution within a one-hour time limit. The abbreviations' meanings are as follows: {\sf Pendulum}(PDL), {\sf Cartpole}(CTP), {\sf Carplatoon}(CPL), {\sf Oscillator}(OSC), {\sf Helicopter}(HLC). }
  \label{fig:compare_with_barrier_certificate}
  
\end{figure}

\section{Related Work}
\label{sec:related}

There has been significant recent interest in exploring techniques to enhance the safety of learning-enabled systems. Techniques described in~\cite{paulsen2020reludiff,singh2019abstract,reluplex,Marabou} define suitable verification methodologies that are capable of providing stronger guarantees for open-loop LECs. For the closed-loop systems discussed in our work, \cite{achiam2017constrained,berkenkamp2017safe,Garcia2015ACS} focus on specifying controller safety as an additional reward. By changing the reward, these methods seek to increase the safety characteristics of the learnt controller. These approaches are different from ours insofar as we consider provably verifiable methods applied independently of a training-based reward framework. \cite{tran2020nnv,huang2019reachnn,verisig,dutta2019reachability,lin2020art} verify a neural network directly.  However, the complexity of the networks, the amount of computation required, and the approximation introduced during the verification process makes these methods difficult to scale to high dimension problems.

Another line of work explores verifiability by applying imitation learning techniques on the subject networks \cite{zhu2019inductive,anderson2020neurosymbolic,yang2021iterative}. These approaches also consider composing a controller family to synthesize a shield. Compared with our work, one significant difference with \cite{zhu2019inductive,anderson2020neurosymbolic,yang2021iterative} is that they choose different controllers based on the system's spatial state. However, in our approach, we select a new controller for every $k$ steps in a trajectory. Hence, our controller selection process is based on temporal behavior. Moreover, since \cite{zhu2019inductive,anderson2020neurosymbolic,yang2021iterative} align a simple imitated controller, which is heavily biased towards safety considerations, with the complex neural controller that also considers performance objectives, scalable verification is challenging, especially when sophisticated performance objectives must be realized.

There also exist tools for synthesizing safe controllers without considering them as shields \cite{sreach,soudjani2015fau}. Similar to \cite{sreach}, our verification algorithm supports linear, time-varying, discrete-time systems that are perturbed by a stochastic disturbance, but our algorithm is demonstrably more scalable. Reinforcement learning algorithms can generally support complex properties defined with various objectives. For example, \cite{camacho2019ltl,toro2018teaching} encode LTL specifications into rewards and train neural controllers with reinforcement learning. However, simply encoding specifications as rewards cannot provide any guarantee on ensuring critical safety properties are preserved. In contrast, our methodology provides desired verifiable results, exploiting the capability of learning other complex properties using standard reinforcement learning techniques. There also exists approaches that consider falsification methods~\cite{dreossi-jar19,dreossi2015efficient,c2e2,deepexplore} that aim to find potential unsafe safes in CPS systems. They can work with complex specifications and high dimensions systems. However, they do not provide provably verifiable guarantees. Another related line of work uses contraction metrics to co-learn controllers and certificates~\cite{sun2020learning}, combining them with Lyapunov certificates. Pursuing this line of work is a topic of future research.

\section{Conclusion}
\label{sec:conclusion}

In this paper, we present a new pipeline that synthesizes a neural network controller with explicit safety guarantees. First, we propose a linear controller family intended to stabilize a system. Then, we verify this family with respect to these safety properties. This verified linear controller family is used in network training and additionally ensures the deployed controller does not violate safety constraints. Because safety verification is decoupled from the training process, our approach has pleasant scalability characteristics that are sensitive to performance objectives. In addition, because we inject the shield into the learning process, the resulting controller is trained with safety considerations in mind, yielding high-quality verified learning-enabled controllers that often outperform their non-verified counterparts. The key insight of our work is that we can decouple properties relevant for learning from those necessary for verification, yielding significant scalability benefits without sacrificing correctness guarantees.

% \section*{Acknowledgement}
% This work was supported in part by C-BRIC, one of six centers in JUMP, a Semiconductor Research Corporation (SRC) program sponsored by DARPA.

\bibliographystyle{splncs04}
\bibliography{ref}

\newpage
\appendix
\gdef\thesection{Appendix \Alph{section}}

\section{Challenge of Distilling Verified Policies}
\label{sec: pendulum_distill_challenging}

\begin{figure}
    \centering
    \includegraphics[width=0.6\linewidth]{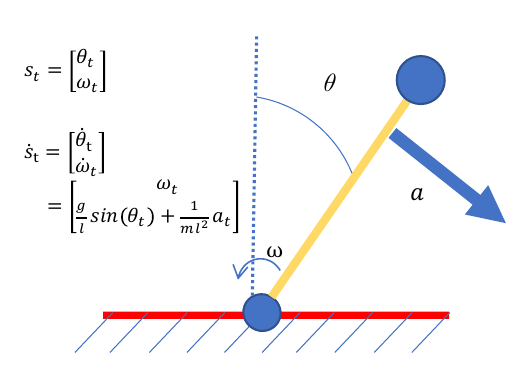}
    \caption{{Demonstration on \sf Pendulum}. We hope to keep the pendulum above the red horizontal line (i.e, $-\frac{\pi}{2} < \theta < \frac{\pi}{2}$), while keep moving with angle velocity that is greater than 0.1 (i.e., $|\omega| > 0.1$).}
    \label{fig:pendulum}
\end{figure}

Previous work~\cite{zhu2019inductive,anderson2020neurosymbolic,yang2021iterative} in shield-based DRL verification distill linear controllers by imitating the neural network, with safety being a primary goal of the learnt controller. In practice, however, reinforcement learning tasks need to consider different non-functional properties and safety. Take the {\sf Pendulum} in Fig.~\ref{fig:pendulum} as an example; we want to simultaneously maximize the velocity of the pendulum while also preserving safety. In Table~\ref{tab:distill_challenge}, we show that a safety controller is distilled from the neural network trained with different rewards that impact their verifiability. We set two different reward functions for the {\sf Pendulum} task in this experiment. In the first setting, we only consider the safety reward (i.e., $\rsafe{s}$); the other considers both liveness and safety (i.e., $\rsafe{s} + \rlive{s}$). For neural network controller trained under these two settings, we distill 50 different linear controllers and verify them with the verification tool provided in \cite{zhu2019inductive}. The $\mathtt{Verified}$ column shows the number of controllers among the 50 distilled controllers considered that were verified to be safe. These results show that when both properties are considered, safety verification becomes significantly more challenging, even for problems as simple as the {\sf Pendulum}.

\begin{table}[h]
    \centering
    \begin{tabular}{@{}lc@{}}
        \toprule
        Reward                 & $\mathtt{Verified}$ \\ \midrule
        Safety rew. only       & $35$                \\   %$34.6  $\pm$  1.36$ \\
        Liveness + safety rew. & $0$                 \\\bottomrule
    \end{tabular}
    \caption{\small The number of controllers that were verified from a total of
        50 distilled under different rewards.}
    \label{tab:distill_challenge}
\end{table}

\section{Verify Linear Controller Family}
\label{sec: verify_linear_controller_family_detailed}

Our verification algorithm aims to maximize the safety probability of the reachable state at each step.  If $\safelb$ is always 1, that is $\srs{t}$ only defines safe elements; then its sum is equal to the total number of steps $M$ in a trajectory.  In this case, the controller is guaranteed to be safe.  We thus seek a  selector $\phiopt$ such that it maximizes the lower bound of $\cumulb = \sum_{t=1}^M \safelb(\lcf, \selector)$: 
\begin{align*} 
    \phiopt = \underset{\selector}{\argmax} \sum_{t=1}^M \safelb(\lcf, \selector) 
\end{align*} 
Suppose that the number of policies in $\lcf$ is $\left| \lcf \right|$.  Finding the optimal $\phiopt$ using a brute-force approach would require traversing all the $\left| \lcf \right|^{\lceil \frac{M}{k} \rceil}$ possible combinations (i.e., all the possible $\selector$) in the worst case.  To improve on this, we consider three pruning strategies. 

First, we keep track of the largest cumulative lower bound of all controllers $\cumulbopt$ for all visited $\selector$. If $\cumulbopt$ reaches $M$, we can terminate immediately, and return the current $\selector$.  We say two selectors have the same prefix at step $t$ if all their selections up to step $t$ are identical.  For a given step $m$, $\ps(\selector, m)$ denotes the set of selectors that have the same prefix as $\selector$. Line 14 to line 17 of Algorithm~\ref{algo:synthesis_phi_opt} describes the first strategy. If there exists a constant $m < M, s.t.\ \left(m - \sum_{t=1}^m \safelb(\lcf, \selector)\right) > (M - \cumulbopt)$, $\ps(\selector, m)$ can be removed from the search space, reducing the number of selectors that need to be considered during verification. 

Second, we store all the $k$-th reachable sets for each $k$ steps. The reachable sets at step $t$ is denoted by $\mathtt{LayerRSet}_t$. Consider two reachable sets $\mathcal{R}^1_k \in \mathtt{LayerRSet}_t$ an $\mathcal{R}^2_k \in \mathtt{LayerRSet}_t$, which are generated by different selectors. Suppose $\mathcal{R}^1_k \subseteq \mathcal{R}^2_k$ and $\selector$ generates $\mathcal{R}^2_k$.  Now, all selectors in $\ps(\selector, k)$ can be removed from the search space. This is because a smaller reachable set is always safer than a larger one. The second strategy corresponds to line 18 to line 26 in the algorithm. Line 20 and line 23 update $\mathtt{LayerRSet}_t$ if any subset relationship is found between the elements of $\mathtt{LayerRSet}_t$. Line 21 and 25 prune the search space $\sss$.

\begin{algorithm}[htbp]
    % \DontPrintSemicolon
    \SetAlgoLined
    \KwIn{$M, \lcf, \sss$}
    \KwOut{$\phiopt, \cumulbopt$}

    $\cumulbopt \leftarrow 0$\;
    \tcp{$\mathtt{LayerRSet}$ stores reachable sets in the same depth}
    Initialize $\mathtt{LayerRSet}_0, \mathtt{LayerRSet}_k, \dots, \mathtt{LayerRSet}_M $ to $\emptyset$

    \For{$\selector \in \sss$}{
        $\cumulb \leftarrow 0$\;
        $\mathtt{PhiRSet} \leftarrow \emptyset$\;
        \For{$i \leftarrow 0$ \KwTo $\lfloor \frac{M}{k}\rfloor$}{
            % \If{$i = \frac{M}{k}$}{break\;}
            $m \leftarrow k (i + 1)$\;
            \For{$j \leftarrow ki$ \KwTo $\min(M, m) - 1$}{
                $\cumulb \leftarrow \cumulb + \hat{p}_j(\lcf, \selector)$\;
            }

            Compute $\srs{m}$\;
            $\mathtt{LayerRSet}_m \leftarrow \mathtt{LayerRSet}_m \cup \{\srs{m}\}$\;
            $\mathtt{PhiRSet} \leftarrow \mathtt{PhiRSet} \cup \{\srs{m}\}$\;

            \tcp{1st strategy, keeping the optimal cumulative lower bound.}
            \If{$m-\cumulb > M-\cumulbopt$}{
                $\sss \leftarrow \sss / \ps(\selector, m)$  \tcp*{Remove $\ps(\selector, m)$ from $\sss$}
                break\;
            }

            \tcp{2nd strategy, cutting among same-layer selectors}
            \For{$\mathcal{R}' \in \mathtt{LayerRSet}_m$}{
                \uIf{$\mathcal{R}' \subset \srs{m}$ }{
                    $\mathtt{LayerRSet}_m \leftarrow \mathtt{LayerRSet}_m / \srs{m}$\;
                    $\sss \leftarrow \sss / \ps(\selector, m)$\;
                }\uElseIf{$\srs{m} \subset \mathcal{R}'$}{
                    $\mathtt{LayerRSet}_m \leftarrow \mathtt{LayerRSet} / \mathcal{R}'$\;
                    Query the $\selector'$ computing $\mathcal{R}'$\;
                    $\sss \leftarrow \sss / \ps(\selector', m)$\;
                }
            }

            \tcp{3rd strategy, invariant of reachable set. }
            \For{$\mathcal{R}' \in \mathtt{PhiRSet}$}{
                \If{$\srs{m} \subseteq \mathcal{R}'$}{
                    Compute $\phi_{inv}$\;
                    \KwRet $\phi_{inv}, M$\;
                }
            }
        }

        \If{$\cumulb > \cumulbopt$}{
            $\cumulbopt = \cumulb$,
            $\phiopt = \selector$\;
        }
        \If{$\cumulb = M$}{
            \KwRet $\phiopt, M$\;
        }
    }
    \KwRet $\phiopt, \cumulbopt$\;
    \caption{Synthesis algorithm for $\phiopt$}
    \label{algo:synthesis_phi_opt}
\end{algorithm}

Finally, we keep the reachable set $\srs{t}$ generated by $\selector$ for every $k$ steps; this set is denoted by $\mathtt{PhiRSet}$ in Algorithm~\ref{algo:synthesis_phi_opt}. $\mathtt{PhiRSet}_t$ is the stochastic reachable set generated by a selector $\selector$ at step $t$. $\mathtt{PhiRSet}_t$ is an invariant set if $\exists t' < t$, $\mathtt{PhiRSet}_t \subseteq \mathtt{PhiRSet}_{t'}$.  At step $t$, if the cumulative lower boundary $\cumulb$ is $t$ (i.e., the probability of safety violation is 0 up to this step), we can return this bound as $M$. Meanwhile, we construct a desired selector $\phi_{inv}$ by keeping the prefix of $\selector$ before step $t$, letting the action of the linear controller at step $t$ govern future steps. This strategy is shown in lines 27 to 32.

The overall structure of the algorithm takes the max simulation step $M$, the linear controller family $\lcf$, and the search space $\sss$ containing all selectors $\selector$ as inputs. It returns the optimal selector $\phiopt$ that maximizes $\cumulbopt = \sum_{t=1}^{M} \safelb \left(\lcf, \phiopt\right)$. The outer for-loop at line 3 traverses all possible selectors.  However, when running, the search space will be pruned, and thus not all the selectors will be visited. As shown in Table~\ref{tab:verification}, we typically can find an optimal selector $\phiopt$ with checking less than $2000$ selectors (which can be estimated with total verification time devided by the time verify per $\lcf$). Considering that the search space containts $10^{5}$ or $10^{10}$ (i.e., $|\lcf|^{\lceil \frac{M}{k} \rceil}$) total combinations, these purning strategies are very efficient.

Given a selector $\selector$, the for-loop at line 6 checks every $k$ steps of the simulation.  $\cumulbopt$ is initialized as 0. If we find a better $\cumulb$ at the end of the loop, it is updated at line 33. From lines 8 to 10, the algorithm computes a cumulative $p_t$ from step $k\cdot i$ to step $\min(M, m) - 1$, and accumulates the safety probability lower bound $\safelb$ for step $k\cdot i$ to step $\min(M, m) - 1$.  Lines 11 to 13 computes $\srs{m}$, and then adds $\srs{m}$ to $\mathtt{LayerRSet}_{m}$ and $\mathtt{PhiRSet}_{m}$. Line 11 computes $\srs{t}$. The parts of the algorithm involving various optimization strategies have been described above.

\section{Benchmark Details}
\label{sec:benchmark_details}

\begin{wrapfigure}{R}{0.5\textwidth}
    \includegraphics[width=\linewidth]{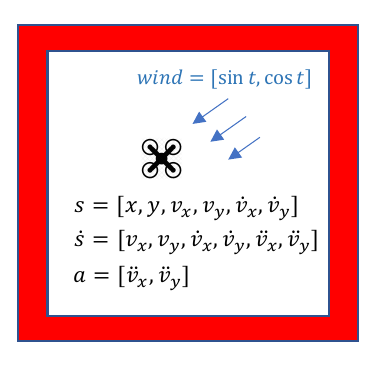}
    \caption{{\sf DroneInWind} benchmark}
    \label{fig:drone_in_wind}
\end{wrapfigure}

Benchmarks {\sf Pendulum}, {\sf Cartpole}, {\sf Carplatoon}, and {\sf Helicopter} are adapted from \cite{DBLP:conf/cav/FanMM018}; {\sf Oscillator} comes from \cite{jordan2007nonlinear}. {\sf DroneInWind} is built by us; an illustrative figure is provided in Fig.~\ref{fig:drone_in_wind}. The state of the drone has 6 dimensions, including position, velocity, and acceleration on a 2D grid. The control signal of the drone is the change of acceleration in the x-y direction. The acceleration is bounded to be smaller than $5 m/s^2$. The drone is not allowed to hit the red wall, and we also limit its speed should be lower than $2 m/s$. Additionally, we require the drone to keep moving within the safety boundary. We model wind with a 2D vector $\left[ \mathtt{sin}\ t, \mathtt{cos}\ t\right]$. Because the wind speed can change the drone's state $s$, we incorporate it into the $A_t$ matrix computing the $\dot{s}$. Because wind changes over time, $A_t$ changes over time as well. As a result, we have a time-variant system. {\sf Pendulum} and {\sf Cartpole} are two classical control models. We have discussed {\sf Pendulum} in detail earlier. {\sf Cartpole} is a control system for a moving cart with a vertical pole; a safety property requires the cart to move without causing the pole to fall. {\sf Carplatoon} models 8 vehicles forming a platoon, maintaining a safe relative distance among one another. {\sf Oscillator} consists of a two-dimensional switched oscillator plus a 16-order filter. The filter smoothens the input signals and has a single output signal; the safety property requires the output to remain below a threshold. {\sf Helicopter} provides a longitudinal motion model of a helicopter; its safety constraint requires that it operate within a specified region.

The other benchmarks in our suite are \emph{stacked} from the first 6 systems but given different safety properties. Perturbations are also added to different stacking elements to yield different behaviors. The prefix number denotes the number of systems stacked. Given a stochastic transition system as defined in Sec.~\ref{sec:stochastic_transition_system}, we stack the $A$, $B$ matrices of the linear control system as diagonal elements of a large matrix. For example, for {\sf 2-Pendulum}, we stack $A, B \in \mathbb{R}^{2 \times 2}$ thus:

\begin{align*}
    A' = \left[\begin{array}{cc}
            A          & \mathbf{0} \\
            \mathbf{0} & A
        \end{array}\right],
    B' = \left[\begin{array}{cc}
            B          & \mathbf{0} \\
            \mathbf{0} & B
        \end{array}\right] \cdot P_B
\end{align*}

Here $A'$ and $B'$ are {\sf 2-Pendulum}'s transition matrices and $A$ and $B$ come from the {\sf Pendulum}. The diagonal elements in the perturbation matrices $P_B \in \mathbb{R}^{4\times4}$ are sampled randomly from the range $[0.95, 1.05]$; all non-diagonal elements are 0.  Because we need to compute $B'\cdot P_B \cdot a_t$, $P_B$ perturbs the input action $a_t$ by a scaling factor.  Similarly, we also stack safety constraints. For example, the safety constraints of {\sf 2-Pendulum} have lower bound $L' \in \mathbb{R}^4$ and upper bound $U' \in \mathbb{R}^4$. The {\sf Pendulum} has lower bound $L \in \mathbb{R}^2$ and upper bound $U \in \mathbb{R}^2$,

\begin{align*}
    L' = P_L \odot \left[ L, L\right], U' = P_U \odot \left[ U, U\right]
\end{align*}

$P_L, P_U \in \mathbb{R}^4$ and their elements are sampled from $[0.95, 1.05]$. Meanwhile, we ensure that every element of $P_L$ is smaller than $P_U$. $\odot$ denotes element-wise multiplication. While we could apply our technique to each component of these stacked systems individually, we evaluate our approach on the high-dimensional case to demonstrate the scalability of our algorithms.

\section{Soundness of Shield}
\label{sec: soundness_of_shield}

\noindent\textbf{Theorem. 1} \textit{(Soundness of Shield) If $\sum_{i=0}^M \safelb = M$,  $s_0 \in \rs{0}$, and $\forall t < M$, $\bar{a}_t$ is generated by Algorithm~\ref{algo:shield}, then $\forall t \leq M, s_t \notin \unsafereg$.}

\begin{proof}
    Supposing the initial state of a system is $s_0$, $\forall s_0 \in \rs{0}$, based on Algorithm~\ref{algo:shield}, if $\region{\hat{s}_{1}} {w} \subseteq \region{\hat{s}'_{1}}{w'} \cup \rs{0}$, the action $a_{nn}$ will be executed. Otherwise, the $a_{safe}$ will be executed. This execution will ensure the next state $s_1$ lands in the reachable set $\srs{1} \cup \rs{0}$, where $\srs{1}$ is the stochastic reachable set of the selected linear controller. Thus,
    \begin{align}
        s_0 \in \rs{0} \implies s_{1} \in \mathcal{R}_{1}^s \cup \rs{0}.
    \end{align} 
    Similarly, for the state $s_t$ at step $t$ and state $s_{t+1}$ at step $t+1$, we have,
    \begin{align}
        \forall t < M, s_t \in \srs{t} \implies s_{t+1} \in \mathcal{R}_{t+1}^s \cup \rs{0}
    \end{align}
    By induction, $\forall t \leq M, s_0 \in \rs{0} \implies s_t \in \bigcup_{i=0}^t\srs{i}$. We select $\phiopt$ that lets $\sum_{i=0}^M \safelb = M$. Thus, all the reachable sets before time step $M$ has no overlapping with $\unsafereg$. $\forall t \leq M, \srs{t} \cap \unsafereg = \emptyset$. Algorithm~\ref{algo:shield} ensures that $\forall t \leq M, s_t \in \bigcup_{i=0}^t\srs{i}$. Thus, 
    \begin{align}
        \forall t \leq M, s_t \in \bigcup_{i=0}^t\srs{i}, \srs{t} \cap \unsafereg = \emptyset\implies s_t \not\in \unsafereg
    \end{align}.
    \qed
\end{proof}

\section{Probabilistic Guarantee Experiments}
\label{sec:probabilistic_guarrantee_exps}
When the noise is large, the system may not be able to be fully verified. Our approach can provide a probabilistic guarantee on the safety lower bound. We demonstrate such guarantee in {\sf Pendulum} and {\sf DroneInWind}, as shown in Table~\ref{tab:probabilistic_guarrantee_exps}.

\begin{table}[!htp]\centering
    \caption{Probabilistic Guarantee Experiments}\label{tab:probabilistic_guarrantee_exps}
    \scriptsize
    \begin{tabular}{l|cc|cc|ccc}\toprule
        \multirow{2}{*}{Benchmarks}        & \multirow{2}{*}{$w'$} & \multirow{2}{*}{$\cumulbopt$} & \multicolumn{2}{c|}{Vio. durining Training $\downarrow$} & \multicolumn{2}{c}{Perf. after Deploying $\uparrow$}                                 \\\cmidrule{4-7}
                                           &                       &                                         & PPO Vio.                                                 & Shield Vio.                                          & Shield/PPO    & Shield/LQR    \\\midrule
        \multirow{2}{*}{{\sf Pendulum}}    & 1.5e-2                & 1                                       & 1437                                                     & \textbf{0}                                           & \textbf{2.65} & \textbf{8.59} \\
                                           & 1.75e-2               & 0.97                                    & 1890                                                     & \textbf{13}                                          & \textbf{3.14} & \textbf{8.67} \\ \midrule
        \multirow{2}{*}{{\sf DroneInWind}} & 2.5e-3                & 1                                       & 864467                                                   & \textbf{0}                                           & \textbf{3.35} & \textbf{4.54} \\
                                           & 3e-3                  & 0.95                                    & 1073107                                                  & \textbf{1076}                                        & \textbf{4.12} & \textbf{5.14} \\
        \bottomrule
    \end{tabular}
\end{table}

As we increase the verification noise $w'$, we show that the probabilistic safety lower bound $\safelb$ (i.e., $\frac{\cumulbopt}{M}$) can drop below $1.0$ on both of the {\sf Pendulum} and {\sf DroneInWind}. As losing the $100\%$ safety guarantee, the safety violation number with shield (i.e., numbers in Shield Vio. column) increases. However, when compared with the safety violation number without the shield (i.e., numbers in PPO Vio. column), the shield still provides much fewer safety violations. We also noticed that losing the safety shield with larger noise slightly increased the performance after deploying. On {\sf Pendulum}, the Shield/PPO increases from $2.65$ to $3.14$; the Shield/LQR increases from $8.59$ to $8.67$.  On {\sf DroneInWind}, the Shield/PPO increases from $3.35$ to $4.12$; the Shield/LQR increases from $4.54$ to $5.14$. 

\section{Probabilistic Reachable Analysis}
\label{sec:prob_reach_analisis}

In this section, we analyze probabilistic reachability. Instead of merely analyzing which states are reachable, we finally provide a probabilistic lower bound for $p_t$. The intuition of the proof is that the multiplication between the upper bound of probability density $f_t(s)$ and the over-approximated unsafe area $\int_{\unsafereg \cap \srs{t}} d s$ is greater than the true unsafe probability. Thus, we can compute the upper bound of unsafe probability. This also gives us the lower bound of safe probability.

\subsection{Probability Density Function of Reachable Set}
\label{apdx:pdf}

We represent the reachable set for the noise-free transition $\nftransmat{t} s_0$ as $\rs{t}$. The reachable set for the stochastic transition at step $t$ is $\srs{t}$ as defined in \ref{sec:stochastic_transition_system}. At begining, the noise-free reachable set is identical to the stochastic reachable set, $\rs{0} = \rs{0}^s$. Because noise increases the size of a reachable set, the noise-free reachable set is subset of the stochastic reachable set, $\rs{t} \subseteq \srs{t}$. 

We hope to calculate the safety probability of stochastic reachable sets on continuous space, thus need to find the Probability Density Function (PDF) for a stochastic reachable set $\srs{t}$. Suppose $f_t(s_t)$ is the PDF of the distribution that $s_t$ subjects to. $s_0 \sim \dist{s_0}$ where $\dist{s_0}$ is a uniform distribution, and the stochastic term $w_0 \sim \dist{w}$. Hence, the reachable distribution of any given step is the linear combination of $\dist{s_0}$ and $\dist{w}$ as we showed in Sec.~\ref{sec:stochastic_transition_system}.

Assuming there is a mapping $r$ from $x$ to $y$, $y = r(x)$, and according to the change of variable formula of PDF,
the PDF of $x$ is $f_x$, the PDF of $y$ is
\begin{align*}
    f_y(y)=f_x\left(x \right)\left|\frac{d}{d y} r^{-1}(y)\right|
\end{align*}

\noindent For the multivariate case, when $x, y \subseteq \mathbb{R}^{n}$,
\begin{align}
    \label{eq:multi_variants_change_of_variable}
    f_y({y})=f_x({x})\left|\det\left(\frac{d }{d {y}} r^{-1}(y)\right)\right|
\end{align}

% \begin{theorem}
%     % https://www.randomservices.org/random/dist/Transformations.html#cov4
%     Suppose that ${X}$ is a random variable taking values in $S \subseteq \mathbb{R}^{n}$, and that ${X}$ has a continuous distribution with probability density function $f_x$. Suppose ${Y}=r({X})$ where $r$ is a differentiable function from $S$ to $T \subseteq \mathbb{R}^{m}$. Then the probability density function of ${Y}$ is given by
%     \begin{align}
%         \label{eq:multi_variants_change_of_variable}
%         f_y({y})=f_x({x})\left|\det\left(\frac{d {x}}{d {y}}\right)\right|, \quad y \in T
%     \end{align}
% \end{theorem}

% \begin{proof}
%     The result follows the multivariate change of variables formula in calculus. If $B \subseteq T$ then
%     $$
%         \mathbb{P}({Y} \in B)=\mathbb{P}[r({X}) \in B]=\mathbb{P}\left[{X} \in r^{-1}(B)\right]=\int_{r^{-1}(B)} f_x({x}) d {x}
%     $$
%     Using the change of variables ${x}=r^{-1}({y}), d {x}=\left|\det\left(\frac{d {x}}{d y}\right)\right| d {y}$ we have
%     $$
%         \mathbb{P}({Y} \in B)=\int_{B} f\left[r^{-1}({y})\right]\left|\det\left(\frac{d {x}}{d {y}}\right)\right| d {y}
%     $$
%     So it follows that $g$ defined in the theorem is a PDF for ${Y}$.
% \end{proof}

Let the noise-free state $\hat{s}_t = \nftransmat{t} s_{0}$ and the noise term $w_t = \nmat{t} w_0$. Suppose the PDF of $\rs{t}$ is $g_t(\hat{s}_t)$ and the PDF of distribution that $\nmat{t} w_0$ subjects to is $h_t(w_t)$. The PDF of $s_{0}$'s distribution and $w_0$'s distribution is $g_0(s_0)$ and $h_0(w_0)$ respectively.

\noindent Applying Eq.~\eqref{eq:multi_variants_change_of_variable} to $g_t$, because $s_0 = \nftransmat{t}^{-1} \hat{s}_t$, $\frac{d s_0}{d \hat{s}_t} = \nftransmat{t}^{-1}$, thus,
\begin{align}
    \label{eq:g_st}
    g_t(\hat{s}_t) & = \frac{g_0\left(\nftransmat{t}^{-1} \hat{s}_t \right)}{\left|\det (\nftransmat{t})\right|} \nonumber \\
                   & = \frac{g_0\left(s_0 \right)}{\left|\det (\nftransmat{t})\right|}
\end{align}

\noindent Applying Eq.~\eqref{eq:multi_variants_change_of_variable} to $h_t$,
\begin{align}
    \label{eq:h_st}
    h_t(w_t) & = \frac{h_0\left((\nmat{t})^{-1} w_t \right)}{\left|\det(\nmat{t})\right|} \nonumber \\
             & = \frac{h_0\left( w_0 \right)}{\left|\det(\nmat{t})\right|}
\end{align}

\noindent Eq.~\eqref{eq:g_st} and Eq.~\eqref{eq:h_st} tell us that the PDF of noise-free state $\hat{s}_t$ and noise $w_t$ can be computed with $g_0(s_0)$ and $h_0(w_0)$, which are the PDF of uniform distributions. The PDF of stochastic state $s_t \in \srs{t}$ is the sum of noise-free state $\hat{s}_t$ and noise $w_t$. The PDF of $s_t$ is $f_t(s_t)$,

\begin{align*}
    f_t(s_t) = g_t(\hat{s}_t) + h_t(w_t).
\end{align*}

% \noindent In the case that will not cause confusion, we remove the subscripts of $s_t$ and $w_t$, and write these two equations as
% \begin{align*}
%     g_t(s) & = \frac{g_0\left(\nftransmat{t}^{-1} s \right)}{\left|\det (\nftransmat{t})\right|}  \\
%     h_t(w) & = \frac{h_0\left((\nmat{t})^{-1} w \right)}{\left|\det(\nmat{t})\right|}
% \end{align*}

The sum of two random variables is distributed as the convolution of their probability densities. Thus, $f_t$ is distributed as the convolution of the distributions $g_t$ and $h_t$. Given a domain $\dist{w} = \left\{w | h_t(w) > 0 \right\}$, $ \dist{\hat{s}} = \{\hat{s} | s \in \srs{t}, (s - \hat{s}) \in \dist{w}\}$; we have that
\begin{align}
    \label{eq:f_st}
    f_{t}(s) = (g_t * h_t) (s) = \int_{\hat{s} \in \dist{\hat{s}}} g_t(\hat{s}) h_t(s - \hat{s})\ d\hat{s}
\end{align}

% \begin{figure}[htp]
%     \centering
%     \includegraphics[width=0.4\linewidth]{images/conv_explaination.pdf}
%     \caption{An illustrating example for Eq.~\ref{eq:f_st} in 2-d case. The $\dist{\hat{s}}$ is computed with $\distw$ and $s$. It can be seen as the result of translating $\distw$ with vector $s$.
%         The PDF of $\rs{t}$ is $g_t(s')$, the PDF of $\dist{\hat{s}}$ is $h_t(s'-s)$. $\dist{\hat{s}}$ acts as filter and moves as $s$ changes.}
%     \label{fig:conv_explaination}
% \end{figure}

Given a reachable set $\srs{t}$, we wish to characterize $p_t$, a measure of how many states of $\srs{t}$ are safe:

$$p_t = \int_{s \in {\bar{\mathcal{S}}_u \cap \srs{t}}} f_{t}(s)\ ds$$

\noindent Here, $f_t(s)$ is the PDF of $\srs{t}$, and it depends on  $g_t(s_0)$ and $h_t(w_0)$. $\bar{\mathcal{S}}_u$ is the state set that satisfies the safety properties. When $\srs{t} \subseteq \bar{\mathcal{S}}_u$, $p_t = 1$

\subsection{Upper Bound of Probability Density Function}
\label{apdx:upper_bound_of_pdf}

Now, we consider the upper bound of $f_t(s)$,

\noindent\textbf{Theorem 1.}\textit{
    The $s_0$ subjects to a uniform distribution on $\left[\lb{s_0}, \ub{s_0}\right]$. $\lb{s_0}, \ub{s_0} \in \mathbb{R}^n$; $n$ is the number of state dimensions. Let $\delta = \ub{s_0} - \lb{s_0}$,
    $$f_t(s) \leq \frac{1}{|det(\nftransmat{t})| \prod_{i=0}^{n-1} \delta_i}$$
}

\begin{proof}
    The $s_0$ subjects to a uniform distribution on $\left[\lb{s_{0}}, \ub{s_{0}}\right]$.
    $$g_0(s_0) = \left\{\begin{array}{cr}
            \frac{1}{\prod_{i=0}^{n-1} \delta_i} & s \in \rs{0} \\
            \\
            0                                    & otherwise
        \end{array}\right.$$
    From Eq. \ref{eq:g_st}, we know that
    \begin{align}
        \label{eq:g_st_final}
        g_t(s_t) = \left\{\begin{array}{cr}
            \frac{1}{|det(\nftransmat{t})| \prod_{i=0}^{n-1} \delta_i} & s \in \rs{t} \\
            \\
            0                                                          & otherwise
        \end{array}\right.
    \end{align}
    $\rs{t}$ is the noise-free reachable set.

    \noindent According to Eq.~\eqref{eq:f_st},
    $$
        f_t(s) = \int_{s' \in \dist{s}} g_t(s') h_t(s - s')\ ds'.
    $$
    From Eq.~\eqref{eq:g_st_final},
    $$g_t(s') \leq \frac{1}{\left|\det\left(\nftransmat{t}\right)\right| \prod_{i=0}^{n-1} \delta_{i}}.$$
    Thus,
    $$
        f_t(s) \leq \frac{1}{\left|\det\left(\nftransmat{t}\right)\right| \prod_{i=0}^{n-1} \delta_{i}} \int_{s' \in \dist{s}} h_t(s' - s)\ ds'
    $$
    $h_t$ is a PDF, thus $\int_{s' \in \dist{s}} h_t(s' - s)\ ds' \leq 1$. We proved that
    \begin{align}
        f_t(s) \leq \frac{1}{|det(\nftransmat{t})| \prod_{i=0}^{n-1} \delta_i}
        \label{eq:general_upper_bound}
    \end{align} \qed
\end{proof}

\noindent The \textbf{Theorem 1} can be extended if the noise is subject to the uniform distribution.

\noindent\textbf{Corollary 1.}\textit{
Suppose the noise on every step subjects to a uniform distribution on $[\epsilon'_l, \epsilon'_h]$; $\epsilon'_{l}, \epsilon'_{h} \in \mathbb{R}^{n}$,
\begin{align*}
    f_{t}(s) \leq \min \left(\frac{1}{\left|\det\left(\mathcal{T}_{s_i}\right)\right| \prod_{i=0}^{n-1} \delta_{i}}, \frac{1}{\left|\det\left(\nmat{t}\right)\right| \prod_{i=0}^{n-1} \delta_{i}^{\prime}}\right),
\end{align*}}
where $\delta' = \epsilon'_h - \epsilon'_l$.

\begin{proof}
    We assume the noise subjects to the uniform distribution on $[\epsilon'_l, \epsilon'_h]$. According to Eq. \ref{eq:h_st}, we have
    $$h_t(s - s') = \left\{\begin{array}{cr}
            \frac{1}{\left|det(\nmat{t})\right| \prod_{i=0}^{n-1} \delta'_i} & s' \in \dist{s} \\
            \\
            0                                                                           & otherwise
        \end{array}\right.$$

    \noindent $g_t$ is PDF, its integration is smaller or equal to 1. According to Eq. \ref{eq:f_st},
    \begin{align*}
        f_t(s) & \leq \frac{1}{\left|\det\left(\nmat{t}\right)\right| \prod_{i=0}^{n-1} \delta'_{i}} \int_{s' \in \dist{s}} g_t(s') \ ds' \\
    \end{align*}
    $g_t$ is a PDF, thus $\int_{s' \in \dist{s}} h_t(s' - s)\ ds' \leq 1$
    \begin{align}
        \label{eq:uniform_noise_upper_bound}
        f_t(s) & \leq \frac{1}{\left|\det\left(\nmat{t}\right)\right| \prod_{i=0}^{n-1} \delta'_{i}}
    \end{align}
    Merge the conclusion in Eq.~\eqref{eq:general_upper_bound} and Eq.~\eqref{eq:uniform_noise_upper_bound},
    \begin{align}
        \label{eq:merged_upper_bound}
        f_{t}(s) \leq \min \left(\frac{1}{\left|\det\left(\nftransmat{t}\right)\right| \prod_{i=0}^{n-1} \delta_{i}}, \frac{1}{\left|\det\left(\nmat{t}\right)\right| \prod_{i=0}^{n-1} \delta_{i}^{\prime}}\right)
    \end{align} \qed
\end{proof}

\subsection{Compute Overapproximation of $\int_{\unsafereg \cap \srs{t}} d s$}
\label{apdx:reachable_approximation}
Since we have the upper bound of $f_t(s)$, if we know which part of $\srs{t}$ violates the safety constraints (i.e., $\unsafereg \cap \srs{t}$), we can integrate the upper bound of $f_t(s)$ on this unsafe part and compute the upper bound of unsafe probability. However, computing exact $\unsafereg \cap \srs{t}$ is difficult as the dimension grows, so we computed the overappoximated $\unsafereg \cap \srs{t}$ with $\srs{t}$'s over-approximation $\mathcal{A}(\srs{t}) \subset \mathbb{R}^n$. The $\mathcal{A}(\srs{t})$ is in the form that $\mathcal{A}(\srs{t}) = \{s| \mathcal{A}_l(\srs{t}) < s < \mathcal{A}_u(\srs{t}) \}$, where $\mathcal{A}_l(\srs{t}), \mathcal{A}_u(\srs{t}) \in \mathbb{R}^n$ are two vectors.

\begin{theorem}
    \label{trm: concretize}
    $\forall \mathcal{T} \in \mathbb{R}^{n \times n}, s \in \mathbb{R}^{n}$,
    \begin{align*}
        \mathcal{T}^{\geq 0} \lb{s} + \mathcal{T}^{< 0} \ub{s} \leq \mathcal{T}s \leq \mathcal{T}^{\geq 0} \ub{s} + \mathcal{T}^{< 0} \lb{s}
    \end{align*}
    Where $\mathcal{T}^{\geq 0}_{ij} = \max(\mathcal{T}_{ij}, 0)$ and $\mathcal{T}^{< 0}_{ij} = \min(\mathcal{T}_{ij}, 0)$. $\lb{s}, \ub{s}$ is the upper and lower boundary of $s$ respectively.
\end{theorem}
\begin{proof}
    Let $s' = \mathcal{T}s$, $s'_i$ is the $i$-th element of $s' \in \mathbb{R}^{n}$. $s$ is the input vector and $s'$ is the output vector.
    $$
        s'_i = \sum_{j=0}^{n-1} \mathcal{T}_{ij} s_j
    $$
    When $\mathcal{T}_{ij} \geq 0$,
    $$\mathcal{T}_{ij} s_j \in [\mathcal{T}_{ij}\lb{s_j}, \mathcal{T}_{ij}\ub{s_j}].$$
    When $\mathcal{T}_{ij} < 0$,
    $$\mathcal{T}_{ij} s_j \in [\mathcal{T}_{ij}\ub{s_j}, \mathcal{T}_{ij}\lb{s_j}].$$
    $\lb{s_j}$ is the lower bound of $s_j$ and $\ub{s_j}$ is the upper bound of $s_j$. Considering the lower bound of output $s'$, when $\mathcal{T}_{ij} \geq 0$, $\mathcal{T}_{ij} s_j \geq \mathcal{T}_{ij} \lb{s_j}$; when $\mathcal{T}_{ij} < 0$, $ \mathcal{T}_{ij} s_j \geq \mathcal{T}_{ij} \ub{s_j}$, as a result,
    $$
        \mathcal{T}s \geq \mathcal{T}^{\geq 0} \lb{s} + \mathcal{T}^{< 0} \ub{s}
    $$
    Similarly, for $\ub{s'}$, when $\mathcal{T}_{ij} \geq 0$, $\mathcal{T}_{ij} s_j \leq \mathcal{T}_{ij} \ub{s_j}$; when $\mathcal{T}_{ij} < 0$, $ \mathcal{T}_{ij} s_j \leq \mathcal{T}_{ij} \lb{s_j}$,
    $$
        \mathcal{T}s \leq \mathcal{T}^{\geq 0} \ub{s} + \mathcal{T}^{< 0} \lb{s}
    $$ \qed
\end{proof}

The stochastic reachable set $\srs{t}$ is stored as $(\nftransmat{t}, \nmat{t}, \lb{s_0}, \ub{s_0}, \lb{w}, \ub{w})$. We can compute its over-approximation $\mathcal{A}(\srs{t})$ with $\nftransmat{t}$ and $\nmat{t}$. First, we compute the over-approximation for the noise-free reachable set $\rs{t}$ with $\nftransmat{t}$. Supposing the initial state space is $[\lb{s_0}, \ub{s_0}]$, according to Theorem~\ref{trm: concretize} and Eq.~\eqref{eq:stochastic_linear_transition},
\begin{align*}
    \mathcal{A}_l(\rs{t}) & = \nftransmat{t}^{\geq 0} \lb{s_0} + \nftransmat{t}^{< 0} \ub{s_0}  \\
    \mathcal{A}_u(\rs{t}) & = \nftransmat{t}^{\geq 0} \ub{s_0} + \nftransmat{t}^{< 0} \lb{s_0}.
\end{align*}
Supposing the noise is bounded by $[\lb{w}, \ub{w}]$, $w$ is a set containing all the possible states of $\nmat{t} w$,

\begin{align*}
    \mathcal{A}_l(w) & = \nmat{t}^{\geq 0} \lb{w} + \nmat{t}^{< 0} \ub{w}, \\
    \mathcal{A}_u(w) & = \nmat{t}^{\geq 0} \ub{w} + \nmat{t}^{< 0} \lb{w}.
\end{align*}

We can compute the $\mathcal{A}(\srs{t})$ by adding $\mathcal{A}(\rs{t})$ and $\mathcal{R}(w)$. Because $\mathcal{A}(\rs{t})$ and $\mathcal{R}(w)$ are two intervals,

\begin{align*}
    \mathcal{A}_l(\srs{t}) & = \mathcal{A}_l(\rs{t}) + \mathcal{A}_l(w) = \nftransmat{t}^{\geq 0} \lb{s_0}+\nftransmat{t}^{<0} \ub{s_0} + \nmat{t}^{\geq 0} \lb{w}+\nmat{t}^{<0} \ub{w} \\
    \mathcal{A}_u(\srs{t}) & = \mathcal{A}_u(\rs{t}) + \mathcal{A}_u(w) = \nftransmat{t}^{\geq 0} \ub{s_0}+\nftransmat{t}^{<0} \lb{s_0} + \nmat{t}^{\geq 0} \ub{w}+\nmat{t}^{<0} \lb{w}
\end{align*} \qed

We assume that the safety constraints are defined as a rectangle. Thus, our following analysis is based on Assumption~\ref{asp: s_u}.

\begin{assumption}
    \label{asp: s_u}
    Given constant lower bound $L \in \mathbb{R}^n$ and constant upper bound $U \in \mathbb{R}^n$ for safe region. The unsafe state set $\unsafereg$ is in the form of
    \begin{align*}
        \unsafereg = \{s \in \mathbb{R}^n | s > U \lor s < L\}
    \end{align*}
\end{assumption}

\begin{theorem}
    \label{trm: int_upper_bound}
    Given lower bound $L \in \mathbb{R}^n$ and upper bound $U \in \mathbb{R}^n$ of the safe region, $\mathtt{clip}(\cdot, L, U)$ bounds the input between the $L$ and $U$.
    \begin{align*}
         & \mathbf{b_1} = \mathcal{A}_u(\srs{t}) - \mathcal{A}_l(\srs{t})                                                                    \\
         & \mathbf{b_2} = \mathtt{clip}(\mathcal{A}_u(\srs{t}), L, U) - \mathtt{clip}(\mathcal{A}_l(\srs{t}), L, U)                          \\
         & \int_{\unsafereg \cap \srs{t}} ds \leq {\prod_{i=0}^{n-1}{(\mathbf{b_1})_i}-\prod_{i=0}^{n-1}{(\mathbf{b_2})_i}}
    \end{align*}
\end{theorem}

Theorem~\ref{trm: int_upper_bound} says the value of $\int_{\unsafereg \cap \srs{t}} ds$ is smaller than its over-approximation $\prod_{i=0}^{n-1}\left(\mathbf{b}_{1}\right)_{i}-\prod_{i=0}^{n-1}\left(\mathbf{b}_{2}\right)_{i}$. An illustrating example is in Fig.~\ref{fig:apply_lower_bound}. $\prod_{i=0}^{n-1}\left(\mathbf{b}_{1}\right)_{i}-\prod_{i=0}^{n-1}\left(\mathbf{b}_{2}\right)_{i}$ represents the area of the yellow frame.

We do not compute the $\int_{\unsafereg \cap \srs{t}} d s$ directly, but compute its upper bound with an over-approximation. Such approximation can be useful when the exact reachable set is expensive to compute in the high dimension case. All the operations can be done with simple matrix operations straightforwardly, which are highly optimized on modern software and hardware.

\subsection{Lower Bound of $p_t$}
\begin{figure}[htb]
    \centering
    \includegraphics[width=0.7\linewidth]{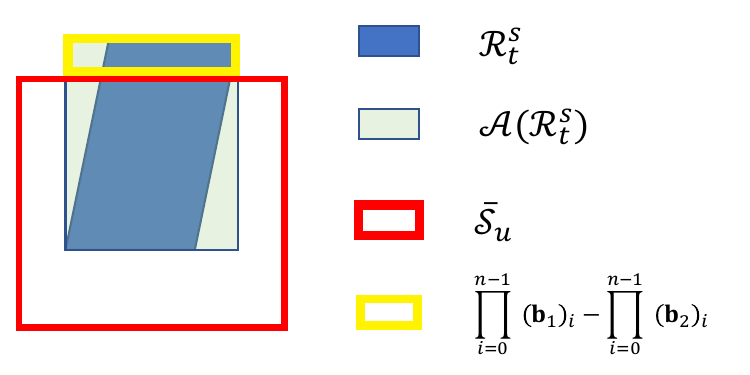}
    \caption{A demo for Theorem~\ref{trm: reachability_low_bound}. $\srs{t}$ is the reachable set at step $t$; $\mathcal{A}(\srs{t})$ is the over-approximation of $\mathcal{R}_s^t$; $\safereg$ is the safe state space; $\prod_{i=0}^{n-1}\left(\mathbf{b}_{1}\right)_{i}-\prod_{i=0}^{n-1}\left(\mathbf{b}_{2}\right)_{i}$ is the area of the yellow frame, and it is used as upper bound of $\int_{\unsafereg \cap \srs{t}} d s$. }
    \label{fig:apply_lower_bound}
\end{figure}

\label{apdx:prob_lb}
With the Theorem~\ref{trm:upper_bound} or Corollary~\ref{trm:uniform_noise_upper_bound}, we can get the upper bound $\rpub$ of $f_t(s)$. $$\rpub = \frac{1}{\left|\det\left(\nftransmat{t}\right)\right| \prod_{i=0}^{n-1} \delta_{i}}$$

\noindent If the noise subjects to uniform distribution,
$$\rpub = \min \left(\frac{1}{\left|\det\left(\nftransmat{t}\right)\right| \prod_{i=0}^{n-1} \delta_{i}}, \frac{1}{\left|\det\left(\nmat{t}\right)\right| \prod_{i=0}^{n-1} \delta_{i}^{\prime}}\right)$$

The intersection between the reachable set $\srs{t}$ and the unsafe state set $\unsafereg$ is ${\unsafereg \cap \srs{t}}$.  Then,
$$
    p_t \geq 1 - \rpub \int_{\unsafereg \cap \srs{t}} ds
$$

\noindent Theorem~{\ref{trm: int_upper_bound}} gives us that $\int_{\unsafereg \cap \srs{t}} d s \leq \prod_{i=0}^{n-1}\left(\mathbf{b}_{1}\right)_{i}-\prod_{i=0}^{n-1}\left(\mathbf{b}_{2}\right)_{i}$. Thus, we get the Theorem~\ref{trm: reachability_low_bound}.

\begin{theorem}
    \label{trm: reachability_low_bound}
    \begin{align*}
        p_t & \geq 1 -  \rpub \cdot \left( \prod_{i=0}^{n-1}\left(\mathbf{b}_{1}\right)_{i}-\prod_{i=0}^{n-1}\left(\mathbf{b}_{2}\right)_{i} \right)
    \end{align*}
\end{theorem}

In Fig \ref{fig:apply_lower_bound}, we provide a demo about how the Theorem \ref{trm: reachability_low_bound} works. Given a stochastic reachable set $\mathcal{R}_s^t$, we compute the over-approximation $\mathcal{A}\left(\srs{t}\right)$. The safety region is defined as $\overline{S}_u$. We can compute the area of $\mathcal{A}(\srs{t}) \cap \unsafereg$ with $\prod_{i=0}^{n-1}{(\mathbf{b_1})_i}-\prod_{i=0}^{n-1}{(\mathbf{b_2})_i}$. Because the probability density function $f_t(s) \leq \rpub$, the cumulative probability for the yellow-wrapped region is upper bounded by $\rpub \cdot (\prod_{i=0}^{n-1}{(\mathbf{b_1})_i}-\prod_{i=0}^{n-1}{(\mathbf{b_2})_i})$. Thus, we can know that the cumulative probability for these safe states is lower bounded by $1 - \rpub \cdot (\prod_{i=0}^{n-1}{(\mathbf{b_1})_i}-\prod_{i=0}^{n-1}{(\mathbf{b_2})_i})$.

\end{document}